\documentclass[aip,cha
,numerical
,reprint
]
{revtex4-1}
\usepackage{amsmath,amsthm,amssymb}
\usepackage{array}
\usepackage{graphicx}
\usepackage{subfig}
\usepackage{caption}
\usepackage{fancyhdr}
\usepackage{color}
\usepackage{extarrows}
\usepackage{enumerate}
\usepackage{epstopdf}
\usepackage{CJK}
\newtheorem{definition}{Definition}

\newtheorem{proposition}{Proposition}

\begin{document}
\begin{CJK*}{GBK}{song}
\title{\textbf{Potential Function in a Continuous Dissipative Chaotic System: Decomposition Scheme and Role of Strange Attractor}}

\author{Yian Ma}
\affiliation{Department of Computer Science and Engineering, Shanghai Jiao Tong University}
\author{Qijun Tan}
\affiliation{School of Mathematical Sciences, Fudan University, Shanghai, 200433, China}
\author{Ruoshi Yuan}
\affiliation{Department of Computer Science and Engineering, Shanghai Jiao Tong University}
\author{Bo Yuan}
\email{boyuan@sjtu.edu.cn.}
\affiliation{Department of Computer Science and Engineering, Shanghai Jiao Tong University}
\author{Ping Ao}
\email{aoping@sjtu.edu.cn.}
\affiliation{Shanghai Center for Systems Biomedicine and Department of Physics\\Shanghai Jiao Tong University, Shanghai, 200240, China}

\date{\today}
\begin{abstract}
In this paper, we demonstrate, first in literature known to us, that potential functions can be constructed in continuous dissipative chaotic systems and can be used to reveal their dynamical properties.
To attain this aim, a Lorenz-like system is proposed and rigorously proved chaotic for exemplified analysis.
We explicitly construct a potential function monotonically decreasing along the system's dynamics, revealing the structure of the chaotic strange attractor.
The potential function can have different forms of construction.
We also decompose the dynamical system to explain for the different origins of chaotic attractor and strange attractor.
Consequently, reasons for the existence of both chaotic nonstrange attractors and nonchaotic strange attractors are clearly discussed within current decomposition framework.
\end{abstract}
\maketitle
\end{CJK*}

\begin{quotation}
Potential function (also known as energy function, generalized Hamiltonian, or Lyapunov function, under different contexts) describes nonlinear dynamical system from a global point of view.
Along this scalar function, all the states in phase space move downward.
Consequently, the potential function accounts for both detailed structure and long term trend of the system's dynamics, indicating its performance and stability at the same time.
It is also anticipated that when the potential function of a system becomes constant, the system has evolved into an ``attractor".
Hence, potential function has special theoretical importance to chaotic system in that it helps reveal the complex structure of chaotic attractor.
Unfortunately, fundamental difficulties pertaining to its construction in nonlinear dynamical systems still need solution.
Failure in its construction has even prompted claims that potential function does not exist in complex systems \cite{strogatz2000nonlinear,rice}.
In this paper, we demonstrate that potential functions can be constructed in chaotic dynamical systems.
We first present a simplified geometrical Lorenz attractor, and rigorously prove that it is chaotic by its Poincar\'e map (which is itself an effort interesting to many researchers \cite{Chen02generateAttractor}).
Then we analytically construct a potential function for the system, accounting for the structure of the chaotic attractor.
With this potential function, we discover that chaotic attractor may not be a strange attractor and vice versa.
This corresponds to previous observations and is explained in detail by virtue of our constructive approach.
\end{quotation}

\section{introduction}
Nonlinear dynamics underlying many natural and technological systems is described by a set of ordinary differential equations:
\begin{align}
    \frac{d\mathbf{x}(t)}{dt}=\mathbf{f}(\mathbf{x}).
\end{align}
This description defines a fixed rule dictating the trend of evolution that a current state will follow into its immediate future.
Overtime, systems with even simple deterministic rules can give rise to seemingly ``fortuitous" phenomena \cite{Poincare}, generally known as chaos.
It remains an intriguing problem as to how these phenomena can be analyzed in a global sense \cite{Strogatz}, under the light of a generic framework in nonequilibrium dynamics \cite{Qian12Stochastic}.
Ideally, such a generic framework should account for: a unified (probably geometric) structure underlying the equations of evolution, a comparable measure of the system's different states, and an accurate reflection of the dynamic process generated by the system.

On an historical account, active search for such a generic framework begins in the 1970s, when Ren\'e Thom and Christopher Zeeman proposed that potential functions and their deformation \cite{Thom} can describe ``the evolution of form in all aspects of nature, and hence it embodies a theory of great generality" \cite{Zeeman}.
Thom's proposal is ingenious.
Because such a scalar function incorporates the previously mentioned attributes into one single quantity: potential function.
It not only generalizes existing approaches of Lyapunov function and first integral, but also encompasses concepts like stability \cite{Sole92Stability} and reversibility \cite{Breymann98Reversibility} into a uniform framework.
Even, natural and technological systems can directly be modeled with potential functions so that the behaviors can still be described ``even when all the internal parameters describing the system are not explicitly known" \cite{Thom}.

However, according to Stephen Smale, Thom's mathematical approach ``deals in with only a few known examples" \cite{Smale}, and hence lacks practical effectiveness.
What Thom and his followers failed to obtain is a proof of existence or even a method of construction for potential functions in complex systems, such as systems with oscillation or chaotic behaviors.
Frustration has even prompted claims that potential functions do not exist in the situations of complex dynamical systems \cite{strogatz2000nonlinear,rice}

Previously, we have rigorously defined potential function in mathematical terms (see definition \ref{def:Lyapunov}) and already demonstrated that potential functions (or Lyapunov functions) can be analytically constructed in oscillating systems \cite{Ao06limitcycle,Yian}.
In this paper, we further motivate such research by showing the construction of potential functions in chaotic systems, and providing additional insights for chaotic and strange attractors.

Actually, constructing potential-like functions in chaotic systems, functions with a restricted part of the properties held by potential functions, is an approach already taken by researchers.
Until very recently, there are still various efforts addressing the issue.
There are generalized Hamiltonian approach \cite{Sira}, energy-like function technique \cite{Sarasola}, minimum action method \cite{Weinan}, and etc., in search for a unified description of chaotic dynamics.
These previous methods all construct a potential-like function to analyze some chaotic system such as the Lorenz system \cite{Lorenz}.
Unfortunately, the scalar functions in these works all lack certain important properties (see section $6$).

Those important drawbacks of the existing methods motivate a real potential function to describe the behavior of chaotic systems.
Ideally, it should be a continuous function in the phase space, monotonically decreasing with time.
Also, when time approaches infinity, potential function should stabilize to some finite quantity if the original system is not divergent.
In this paper, we combine these ideas together as a potential function and try to obtain it in chaotic systems.

The paper is organized as follows.
First of all, we define an ideal potential function for the description of chaotic systems.
To analyze chaotic systems in detail with this potential function, we create an attractor that is chaotic by definition.
Then, potential function for this chaotic attractor is constructed, showing the structure of the chaotic strange attractor.
In addition, our framework provides a decomposition of the original vector field.
The decomposition helps understand the different origins for chaotic attractor and strange attractor, explaining why there exists both chaotic nonstrange attractors and nonchaotic strange attractors.

\section{potential function}
We first state the definition of a potential function, which is a natural description of dynamical systems with monotonic properties.
Then we will discuss a decomposition scheme of generic dynamical systems associated with the potential function.

\begin{definition}[Potential Function \cite{ao04potential,Ruoshi}]
\label{def:Lyapunov}

  Let $\Psi:\mathbb{R}^n\xrightarrow{}\mathbb{R}$ be a continuous function. Then $\Psi$ satisfying the following condition is called a potential function for the dynamical system $\dot{\mathbf{x}}=\mathbf{f}(\mathbf{x}):\mathbb{R}^n\xrightarrow{}\mathbb{R}^n$.
  \begin{enumerate}[(a)]
  \item
  $\dot{\Psi}(\mathbf{x})=d\Psi/dt|_\mathbf{x}\leqslant0$ for all $\mathbf{x}\in \mathbb{R}^n$  if  $\dot{\Psi}(\mathbf{x})$ exists.

  \item
  $\nabla\Psi(\mathbf{x^*})=0$ if and only if $\mathbf{x}^*\in \mathbb{A}$, where $\mathbb{A}$ is the attractor of the dynamical system: $\dot{\mathbf{x}}=\mathbf{f}(\mathbf{x})$.
  \end{enumerate}
\end{definition}

Here, an attractor $\mathbb{A}$ of a dynamical system with flow $\phi_t$ can be fixed point, limit cycle, or chaotic attractor.
In this sense, a potential function is a Lyapunov function.
An attractor $\mathbb{A}$ is formally defined as the following.

\begin{definition}[Attractor]
\label{def:Attractor}

  An attractor $\mathbb{A}$ of a dynamical system $\dot{\mathbf{x}}=\mathbf{f}(\mathbf{x})$ with flow $\phi_t$ is a compact invariant set, with an open set $U$ containing $\mathbb{A}$ such that for each $\mathbf{x}\in U$, $\phi_t(\mathbf{x})\in U$ for all $t\geqslant0$ and $\mathbb{A}=\bigcap_{t\geqslant0}\phi_t(U)$.
\end{definition}

This definition balances between different fashions of literatures \cite{hirsch,Shilnikov} and is exactly the same as the definition of an ``attracting set" in a classical textbook of dynamical systems \cite{robinson}.

\subsection{Decomposition Scheme}
Many efforts generalizing Hamiltonian or gradient dynamics realize that there generally exists a decomposition of dynamical systems $\dot{\mathbf{x}}=\mathbf{f}(\mathbf{x})$ \cite{Ruoshi,Daizhan,Wangjin}:
\begin{align}
\dot{\mathbf{x}}=\mathbf{f}(\mathbf{x})=M(\mathbf{x}) \nabla\Psi(\mathbf{x}).\nonumber
\end{align}

And $M(\mathbf{x})$ can be decomposed into:
\begin{align}
M(\mathbf{x})=J(\mathbf{x})-D(\mathbf{x})+Q(\mathbf{x}),\nonumber
\end{align}
where $J(\mathbf{x})$ and $D(\mathbf{x})$ are semi-positive definite symmetric matrices and $Q(\mathbf{x})$ is skew-symmetric.

In our definition of the potential function $\Psi(\mathbf{x})$, however, matrix $J(\mathbf{x})$ do not appear.
That is:
\begin{align}
\dot{\mathbf{x}}&=\mathbf{f}(\mathbf{x})\\\nonumber
&=-D(\mathbf{x}) \nabla\Psi(\mathbf{x})+Q(\mathbf{x}) \nabla\Psi(\mathbf{x}).\label{decomp}
\end{align}

This means that a generic system is only composed of an energy dissipating (gradient) part and an energy conserved (rotation) part.
For the gradient part, potential $\Psi$ is a common energy function; for the rotation part, $\Psi$ is a first integral.
And we can express these two parts once we find the potential function for the system \cite{Ruoshi}:
\begin{align}
    D=-\frac{\mathbf{f}\cdot\nabla\Psi}{\nabla\Psi\cdot\nabla\Psi}I,
\end{align}
and
\begin{align}
    Q=\frac{\mathbf{f}\times\nabla\Psi}{\nabla\Psi\cdot\nabla\Psi}.
\end{align}
Here, $I$ denotes identity matrix and the generalized cross product of two vectors defines a matrix:
$\mathbf{x} \times \mathbf{y}=A=(a_{ij})_{n\times n}=(x_iy_j-x_jy_i)_{n\times n}$.

Clearly, $\left(D \nabla\Psi(\mathbf{x})\right)\times\nabla\Psi=0$ and $\left(Q \nabla\Psi(\mathbf{x})\right)\cdot\nabla\Psi=0$, corresponding exactly to the curl-free component and divergence-free component in Helmholtz decomposition \cite{Helmholtz}.

Further, as have been discussed in the context of nonequilibrium thermal dynamics \cite{Ao07Decomp}, such two parts correspond to two different structures in geometry:
a dissipative bracket $\{\cdot,\cdot\}$; and a generalized Poisson bracket \cite{Arnold89Math} $[\cdot,\cdot]$ \footnote{To avoid confusion, we restrict the use of generalized Poisson brackets in this section (section $2$)}.

A dissipative bracket $\{\cdot,\cdot\}$ associates with matrix $D(\mathbf{x})$:
\begin{align}
\{f, g\}=\partial_i f D_{ij} \partial_j g,\nonumber
\end{align}
and is generally defined as symmetric: $\{f, g\}=\{g, f\}$;
and semi-positive definite: $\{f, f\}\geqslant 0$;
satisfying Leibniz' rule: $\{fg , h\}=f\{g , h\}+g\{f , h\}$.

While a generalized Poisson bracket $[\cdot,\cdot]$ associates with matrix $Q(\mathbf{x})$:
\begin{align}
\left[f, g\right]=\partial_i f Q_{ij} \partial_j g,\nonumber
\end{align}
and can be generally defined as antisymmetric: $\left[f, g\right]=-\left[g, f\right]$;
satisfying Leibniz' rule: $\left[fg , h\right]=f[g , h]+g[f , h]$.


Hence, the original differential equations can be expressed as:
\begin{align}
\dot x_i=-\{x_i,\Psi\}+\left[x_i,\Psi\right].\nonumber
\end{align}
That is, a generic dynamical system is a direct composition of the two well-studied geometric structures.

Later, this gradient-rotation decomposition would provide additional insight to the understanding of chaotic attractors and strange attractors.

\section{Simplified Geometric Lorenz Attractor}
As can be seen in previous works, many efforts have been made to analyze Lorenz system \cite{Lorenz} as a typical model for chaos.
Yet, to the best knowledge of the authors, there is only numerical evidence \cite{Binder99Numeric} that the Lorenz equations support a robust strange attractor \cite{Tucker}.
Total understanding of the Lorenz attractor, including but not limited to an analytic proof that the Lorenz attractor is chaotic is still lacking \cite{SmaleCentury}.

An early work \cite{John} attempted to study chaotic systems by constructing a geometric model in a piecewise fashion to resemble the Lorenz system.
The resultant ``geometric Lorenz attractor" from the piecewise model is studied in some depth and an analogy is made between it and the Lorenz system \cite{Tucker}.
This methodology is practically effective, yet the model system can become even simpler to be analytically proved as chaotic.

Hence, we start out constructing a simplified geometric Lorenz attractor.
The model system is described by piecewise continuous ordinary differential equations (ODE), similar to the ``geometric Lorenz attractor".
We integrate trajectories in each continuous region of the model system.
Then we reveal the structure of the attractor by finding the Poincar\'e map between the continuous regions.
Through the Poincar\'e map, the attractor is proved to be a chaotic attractor according to the widely applied definition \cite{robinson} of Devaney chaos.

\begin{figure}
\begin{center}
  \includegraphics[width=0.5\textwidth]{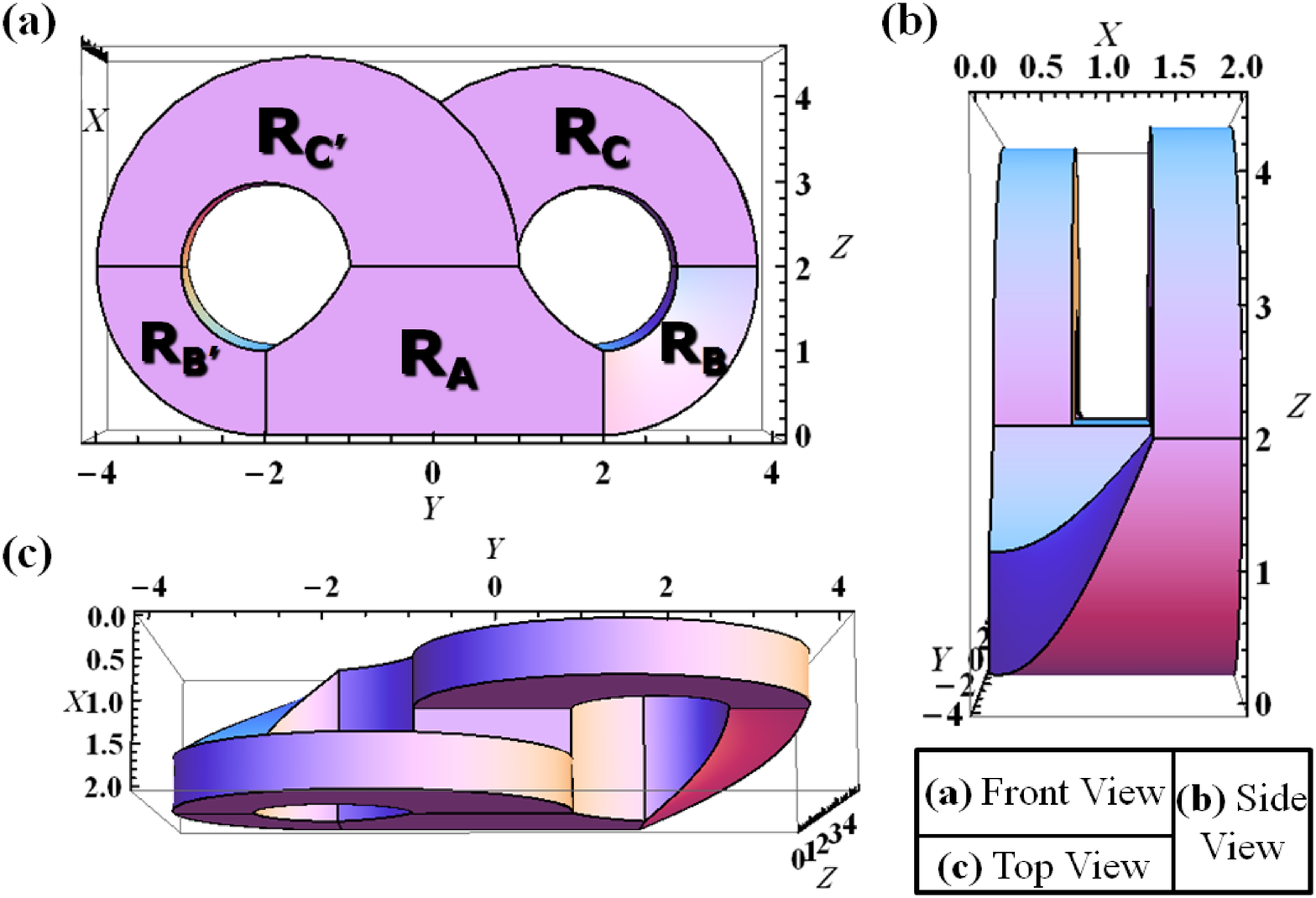}
\end{center}
\caption{
\begin{flushleft}
\textbf{Figure 1} $|$ \textbf{The Simplified Geometric Lorenz Attractor.}
The dynamical system we study here is defined piecewisely in region $R_A$, $R_B$, and $R_C$, along with region $R_{B'}$ and $R_{C'}$ as the symmetric counterparts of $R_B$ and $R_C$.
The front, side, and top view of the regions of definition are shown in panel \textbf{(a)} through panel \textbf{(c)} respectively.
Trajectories of this dynamical system would converge into an attractor $\mathbb{A}_L$ (see subsection D and figure $7$), which is a simplified version of the geometric Lorenz attractor \cite{John}.
\end{flushleft}
}
\label{fig:vector}
\end{figure}

\subsection{Model System Description}
The piecewise continuous ODE model is described in each continuous region (from $R_A$ to $R_C$, along with $R_{B'}$ and $R_{C'}$ as the symmetric counterparts of $R_B$ and $R_C$) as follows, corresponding to Figure ($1$).

\begin{enumerate}[1.]
\item In region $R_A$, where $x\in[0,2]$ \footnote{The square brackets in this (section $3$) and the following sections mean closed intervals, not the generalized Poisson brackets}, $y\in[-2,2]$, $z\in[0,2]$, $yz\in[-2,2]$:
\begin{equation}
\left\{
    \begin{array}{l}
     \dot x=0\\
     \dot y=y\\
     \dot z=-z.
     \end{array}
\right.
\end{equation}

Dynamics in this region is characterized by saddle points at $y=z=0$.
These saddle points are responsible for causing bifurcation in originally close trajectories.

\item Region $R_B$ is defined as: $x\in[0 \, , \; 2/3+8/(3\pi)\times \theta]$, $y\in(2,4]$, $z\in[0,2]$, $\sqrt{(z-2)^2+(y-2)^2}\in[1,2]$,
where:
\begin{align}
    \theta=\arccos{\dfrac{y-2}{\sqrt{(z-2)^2+(y-2)^2}}},\nonumber
\end{align}
denoting the angle that point $(y,z)$ form with respect to the center $(2,2)$.

In $R_B$:
\begin{equation}
\left\{
    \begin{array}{l}
     \dot x=-\dfrac{x}{\theta+\pi/4}\\
     \dot y=2-z\\
     \dot z=y-2.
     \end{array}
\right.\label{regionB}
\end{equation}

Trajectories in this region rotate for an angle of $\pi/2$ with respect to $y=z=2$ and contract in the $x$ direction.

\item In region $R_C$, where $x\in[0, \, 2/3]$, $y\in[-1,4]$, $z>2$, $\sqrt{(z-2)^2+(y-2)^2}\geqslant 1$, $\sqrt{(z-2)^2+(y-3/2)^2}\leqslant 5/2$:
\begin{equation}
\left\{
    \begin{array}{l}
     \dot x=0\\
     \dot y=2-z\\
     \dot z=\dfrac{9y}{8}-\dfrac{21}{8}+\dfrac{\sqrt{(3y-7)^2+8(z-2)^2}}{8}.
     \end{array}
\right.
\end{equation}

In this region, trajectories rotate for another angle of $\pi$ with respect to $y=z=2$ and expand in the $y$ direction.

The whole system is set symmetrical with respect to the line: $x=1$; $y=0$; $z\in\mathbb{R}$.
So we change the coordinate of $(x,y,z)$ into $(2-x,-y,z)$ to have expressions of the vector field in region $R_{B'}$ and $R_{C'}$ from expressions in region $R_B$ and $R_C$.

\item Region $R_{B'}$ is defined as: $x\in[4/3-8/(3\pi)\times \theta \, , \; 2]$, $y\in[-4,-2)$, $z\in[0,2]$, $\sqrt{(z-2)^2+(y+2)^2}\in[1,2]$,
where:
\begin{align}
    \theta=\arccos{\dfrac{-y-2}{\sqrt{(z-2)^2+(y+2)^2}}},\nonumber
\end{align}
denoting the angle that point $(y,z)$ form with respect to the center $(-2,2)$.

In $R_{B'}$:
\begin{equation}
\left\{
    \begin{array}{l}
     \dot x=\dfrac{2-x}{\theta+\pi/4}\\
     \dot y=z-2\\
     \dot z=-y-2.
     \end{array}
\right.
\end{equation}

Vector field in region $R_{B'}$ corresponds exactly to that in $R_B$.

\item In region $R_{C'}$, where $x\in[4/3, \, 2]$, $y\in[-4,1]$, $z>2$, $\sqrt{(z-2)^2+(y+2)^2}\geqslant 1$, $\sqrt{(z-2)^2+(y+3/2)^2}\leqslant 5/2$:
\begin{equation}
\left\{
    \begin{array}{l}
     \dot x=0\\
     \dot y=z-2\\
     \dot z=-\dfrac{9y}{8}-\dfrac{21}{8}+\dfrac{\sqrt{(3y+7)^2+8(z-2)^2}}{8}.
     \end{array}
\right.
\end{equation}

Vector field in region $R_{C'}$ corresponds exactly to that in region $R_C$.

We note that the model system in region $R_{B'}$ and $R_{C'}$ is just a change of variables of the system in region $R_B$ and $R_C$.
So, to avoid redundancy, we will only take region $R_A$, $R_B$ and $R_C$ to represent all the regions of definition in the following analysis.
\end{enumerate}

\subsection{Near Saddle-Focus Fixed Points}
We have constructed the model system containing one saddle fixed point.
As in the Lorenz system, there would actually be another two saddle-focus fixed points when the system expands to the whole $\mathbb{R}^3$ space.
Here, we complete the dynamical system near the two saddle-focus fixed points so that the convergence behavior away from the attractor can be further demonstrated.

We denote the regions of definition discussed here as region $R_D$ and $R_{D'}$ (see Figure \ref{fig:AttractorCenter}), each consisting of three parts: region $R_{D^A}$, $R_{D^B}$, and $R_{D^C}$ (for region $R_D$ as example).
Region $R_D$ and $R_{D'}$ are symmetrical with respect to the line: $x=1$; $y=0$; $z\in\mathbb{R}$, just as in the previous section.
Hence, we follow the convention stated in the previous section: to take region $R_{D^A}$, $R_{D^B}$, and $R_{D^C}$ representing their symmetrical counterparts.
The regions: $R_{D^A}$, $R_{D^B}$, and $R_{D^C}$ and the differential equations in them are written as the following.

\begin{figure}
\begin{center}
  \includegraphics[width=0.5\textwidth]{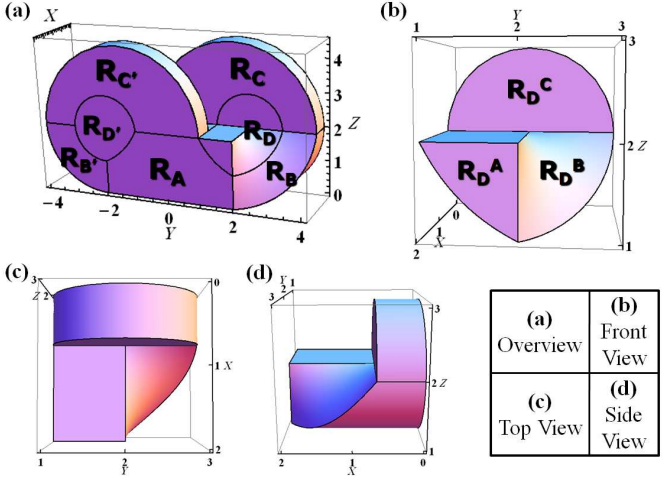}
\end{center}
\caption{
\begin{flushleft}
\textbf{Figure 2} $|$ \textbf{Near Saddle-Focus Fixed Points.}
When the system expand to contain region $R_D$ and $R_{D'}$, two saddle-focus fixed points would emerge.
This figure elaborates on the system near the two saddle-focus fixed points.
\textbf{(a)}, The regions containing the two saddle-focus fixed points, i.e. $R_D$ and $R_{D'}$ are shown along with other regions.
\textbf{(b-d)}, The front, top, and side view of region $R_D$ are shown respectively.
\end{flushleft}
}
\label{fig:AttractorCenter}
\end{figure}

\begin{enumerate}[1.]
\item Region $R_{D^A}$ is close to region $R_A$ and is defined as: $x\in[0,2]$, $y\in[1,2]$, $z\in[1,2]$, $yz\in(2,4]$.
We simply take differential dynamical system in it the same as that in region $R_A$:
\begin{equation}
\left\{
    \begin{array}{l}
     \dot x=0\\
     \dot y=y\\
     \dot z=-z.
     \end{array}
\right.
\end{equation}

Hence, states in this region are unstable in the $y$ direction and stable in the $z$ direction, causing a rotation effect.

\item Region $R_{D^B}$ is close to region $R_B$ and is defined as: $x\in[0 \, , \; 2/3+8/(3\pi)\times \theta]$, $y\in(2,3]$, $z\in[1,2]$, $\sqrt{(z-2)^2+(y-2)^2}<1$.

Here,
\begin{align}
    \theta=\arccos{\dfrac{y-2}{\sqrt{(z-2)^2+(y-2)^2}}},\nonumber
\end{align}
denoting the angle that point $(y,z)$ form with respect to the center $(2,2)$.

In region $R_{D^B}$:
\begin{equation}
\left\{
    \begin{array}{l}
     \dot x=-\dfrac{x}{\theta+\pi/4}\\
     \dot y=2-z\\
     \dot z=y-2.
     \end{array}
\right.\label{regionB}
\end{equation}

Same as in region $R_B$, trajectories in this region rotate for an angle of $\pi/2$ with respect to $y=z=2$ and contract in the $x$ direction.

\item In region $R_{D^C}$, where $x\in[0, \, 2/3]$, $y\in[1,3]$, $z>2$, $\sqrt{(z-2)^2+(y-2)^2}<1$:
\begin{equation}
\left\{
    \begin{array}{l}
     \dot x=0\\
     \dot y=2-z+\left(y-2\right)\left(\dfrac{1}{\sqrt{(z-2)^2+(y-2)^2}}-1\right)\\
     \dot z=y-2+\left(z-2\right)\left(\dfrac{1}{\sqrt{(z-2)^2+(y-2)^2}}-1\right).
     \end{array}
\right.
\end{equation}

In this region, trajectories tend to converge to the unit-radius circle centered at $y=z=2$.
Hence, states in the whole region $R_{D}$ are attracted to the circle: $x=0, \, \sqrt{(z-2)^2+(y-2)^2}=1$.

\end{enumerate}

Here, it is observable that region $R_D$, containing a saddle-focus fixed point, form a semi-stable limit cycle at $x=0, \, \sqrt{(z-2)^2+(y-2)^2}=1$ (when $y, z \in [1,2]$, the curve of the limit cycle changes expression to: $x=0, \, yz=2$).
This limit cycle locates at the boundaries between region $R_D$ and its adjacent regions.
This phenomenon corresponds with many observations that fixed points transit into chaotic behaviors through limit cycles \cite{Weinan}.

Also, this section demonstrates that the domain of definition in the model system is not restricted to the regions discussed above.
If we take dynamical system in the rest of $\mathbb{R}^3$ space converging into the defined regions ($R_A$ through $R_D$, $R_{B'}$ through $R_{D'}$), domain of definition can be expanded to the whole space.

\subsection{Trajectory and Poincar\'e Map}

\begin{figure}
\begin{center}
  \includegraphics[width=0.5\textwidth]{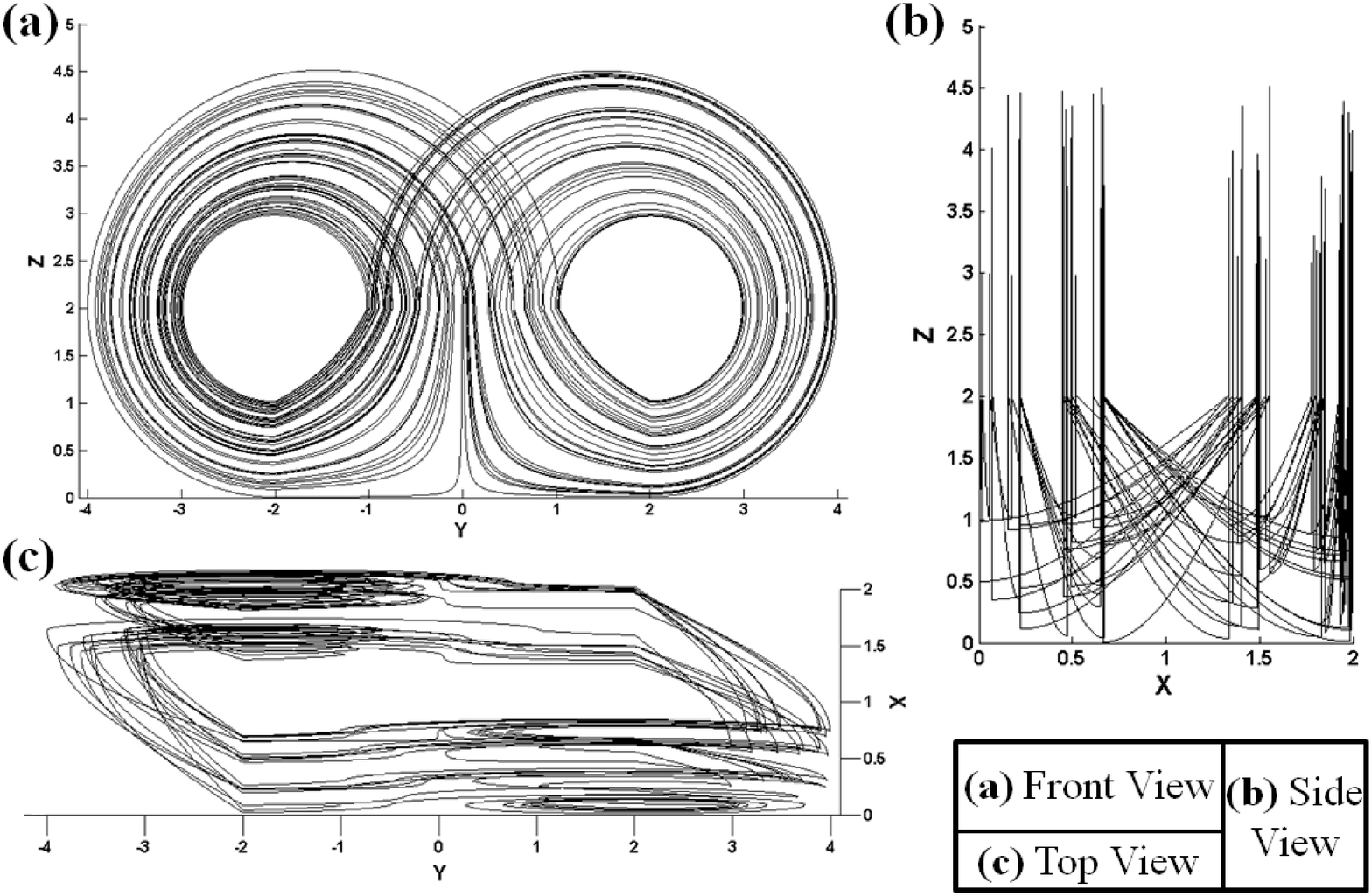}
\end{center}
\caption{
\begin{flushleft}
\textbf{Figure 3} $|$ \textbf{Simulated Trajectory of the System.}\\
We simulated a trajectory of the dynamical system constructed.
It appears to have ``erratic" behaviors, similar to the Lorenz system.
\end{flushleft}
}
\label{fig:simulation}
\end{figure}

Based on Equation (5-9), we simulate the trajectory of the dynamical system (shown in Figure (\ref{fig:simulation})).
Actually, trajectories in each region can be analytically solved.
To study the structure of the attractor, we solve the trajectories in region $R_A$, $R_B$, and $R_C$ respectively:
\begin{enumerate}[1.]
\item In region $R_A$, trajectories are represented as:
\begin{equation}
\left\{
    \begin{array}{l}
     x=x_0\\
     y=y_0 e^t\\
     z=z_0 e^{-t},
     \end{array}
\right.
\end{equation}
where $z_0$ can usually be taken as $2$.

Hence, states in this region are exponentially unstable in the $y$ direction and exponentially stable in the $z$ direction.

\item In region $R_B$, trajectories are:
\begin{equation}
\left\{
    \begin{array}{l}
     x=x_0\left(\dfrac{1}{3} - \dfrac{4}{3\pi} t \right)\\
     y=\sqrt{y_0^2+z_0^2} \cos{t} +2\\
     z=\sqrt{y_0^2+z_0^2} \sin{t} +2,
     \end{array}
\right.\label{trajB}
\end{equation}
where $y_0$ can be $2$.

Here, we can observe that $x(t)$ decreases monotonically while $y(t)$ and $z(t)$ form a circle.


\item In region $R_C$, trajectories are:
\begin{equation}
\left\{
    \begin{array}{l}
     x=x_0\\
     y=\sqrt{y_0^2+z_0^2} \cos{t} + \dfrac{1}{3}\left(1 - \sqrt{y_0^2+z_0^2}\right) + 2\\
     z=\sqrt{y_0^2+z_0^2} \sin{t} + 2,
     \end{array}
\right.
\end{equation}
where $z_0$ can be $2$.

Hence, trajectories in region $R_C$ move along circles determined by initial conditions.

\end{enumerate}

Then, to further study the structure of the attractor of the system, we want to calculate the Poincar\'e map of the system.
Here, we take Poincar\'e surface of section as: $z=2$, and find the resultant Poincar\'e map as a discrete dynamical system defined on $[0,2]\times[-1,1]$ (shown in Figure (\ref{fig:poincare}) and follows).


\begin{figure}
\begin{center}
  \includegraphics[width=0.5\textwidth]{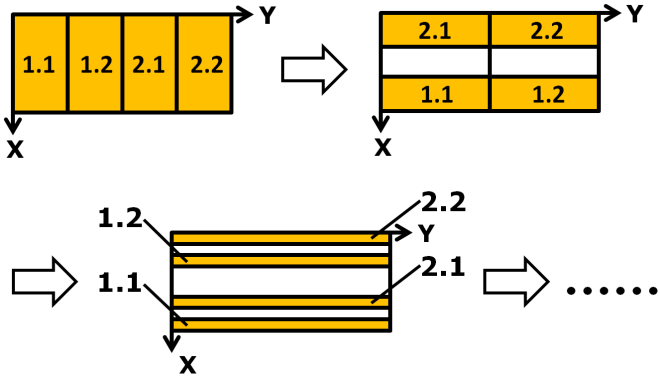}
\end{center}
\caption{
\begin{flushleft}
\textbf{Figure 4} $|$ \textbf{Poincar\'e Map of the Dynamical System.}\\
This Poincar\'e map is taken over the surface of $z=2$.
It can readily be observed that this map is a ``baker's map" \cite{Kuznetsov} and creates a Cantor set multiplying a real line segment.
\end{flushleft}
}
\label{fig:poincare}
\end{figure}

When $(x,y)\in [0,2]\times[0,1]$,
\begin{equation}
\left\{
    \begin{array}{l}
     x_{n+1}=\dfrac{1}{3}x_n\\
     y_{n+1}=2y_n-1;
     \end{array}
\right.
\end{equation}

When $(x,y)\in [0,2]\times[-1,0)$,
\begin{equation}
\left\{
    \begin{array}{l}
     x_{n+1}=\dfrac{1}{3}x_n+\dfrac{4}{3}\\
     y_{n+1}=2y_n+1.
     \end{array}
\right.
\end{equation}

The above discrete dynamical system is just a ``baker's map" defined in literatures before \cite{Kuznetsov}, and can be understood figuratively as the following.
In the $x$ direction, the mapping is contractive.
The square of definition: $[0,2]\times[-1,1]$ is contracted to the one third of it, forming a rectangle: $[0,2/3]\times[-1,1]$.
In the $y$ direction, the mapping is expansive just as the doubling map \cite{robinson}.
The rectangle is stretched to: $[0,2/3]\times[-3,1]$.
Then we keep the right half of the resulted rectangle and move the left half: $[0,2/3]\times[-3,-1)$ to the position: $[0,2/3]\times[-1,1)$.
It can readily be seen that the invariant set is formed by iteratively removing the middle third of the intervals along the $x$ direction.

\subsection{Attractor of the Model System}
We denote the attractor of the model system as: $\mathbb{A}_L$.
And with the Poincar\'e map of the model system defined as a dynamical system on $[0,2]\times[-1,1]$ (in the previous section), we denote its attractor as: $\mathbb{A}_P$.
In this section, we first express attractor $\mathbb{A}_P$ of the Poincar\'e map in terms of the Cantor set; then we express attractor $\mathbb{A}_L$ of the original model system.



We have already found the Poincar\'e map of the model system as iteratively removing the middle third of the invariant sets along the $x$ direction.
That is:
first, remove the set $(2/3,4/3)\times[-1,1]$;
then, remove the middle third of the left two sets $[0,2/3]\times[-1,1]$ and $[4/3,2]\times[-1,1]$;
and iterate the process all along.
We present all the removed intervals iteratively as the following:

\begin{align}
    \mathbb{C}_1=\left(\frac{2}{3},\frac{4}{3}\right)\times\left[-1,1\right],\nonumber
\end{align}
and
\begin{align}
    \mathbb{C}_{n+1}=\left(\frac{\mathbb{C}_n}{3}\bigcup\frac{\mathbb{C}_n+4}{3}\right)\times\left[-1,1\right].
\end{align}

Then the attractor of the Poincar\'e map is $[0,2]\times[-1,1]$ minus the union of all the sets $\mathbb{C}_i$:
\begin{align}
    \mathbb{A}_P&=\left[0,2\right]\times\left[-1,1\right]-\bigcup_{i=1}^{\infty}\mathbb{C}_i\times\left[-1,1\right]\nonumber\\
    &=\left(\left[0,2\right]-\bigcup_{i=1}^{\infty}\mathbb{C}_i\right)\times\left[-1,1\right]\nonumber\\
    &=\mathbb{C}\times\left[-1,1\right],
\end{align}
where $\mathbb{C}$ denotes the Cantor set \cite{frontier} defined on the interval of $[0,2]$.

Hence, attractor $\mathbb{A}_P$ of the Poincar\'e map is the Cantor set multiplying a real line segment.
We can calculate its box-counting dimension \cite{frontier} to be:
\begin{align}
d_b(\mathbb{A}_P)=\lim_{\epsilon\rightarrow0}\dfrac{\log N(\epsilon,\mathbb{A}_P)}{\log(1/\epsilon)}=1+ln(2)/ln(3).
\end{align}
Hence, the attractor $\mathbb{A}_P$ is of fractal dimension, a strange attractor \cite{NonStrangeChaotic}.

With the trajectories of the system analytically solved in each region, we further express attractor $\mathbb{A}_L$ of the model system as
(assuming $\theta=\arccos{(y-2)/\sqrt{(z-2)^2+(y-2)^2}}$):

In region $R_A$, $x\in\mathbb{C}$;

in region $R_B$, $\left(4/(3\pi)\times\theta+1/3\right)^{-1}x\in\mathbb{C}$;

in region $R_C$, $x\in\mathbb{C}$.

The box-counting dimension of attractor $\mathbb{A}_L$ is then calculated to be:
\begin{align}
d_b(\mathbb{A}_L)=\lim_{\epsilon\rightarrow0}\dfrac{\log N(\epsilon,\mathbb{A}_L)}{\log(1/\epsilon)}=2+ln(2)/ln(3).
\end{align}
It is hence a strange attractor with fractal dimension.



It will be proved in the following section that attractors $\mathbb{A}_P$ and $\mathbb{A}_L$ are also chaotic attractors.

\subsection{Proof of the Attractor as Chaotic}
By the widely applied definition \cite{robinson} of Devaney chaos, an attractor $\mathbb{A}$ is defined as a chaotic attractor if:
\begin{enumerate}[1.]
\item
the attractor is indecomposable (i.e., if $\emptyset\neq \mathbb{A}'\subseteq \mathbb{A}$ is an attractor, then $\mathbb{A}'=\mathbb{A}$);
\item
the system is sensitive to initial conditions when restricted to $\mathbb{A}$ (defined in the following definition \ref{def_sensitive}).
\end{enumerate}


\begin{definition}
    A map (a continuous-time system is defined similarly) has sensitive dependence on initial conditions when restricted to its invariant set $\mathbb{A}$, if there exists $r$, for any $p_0 \in \mathbb{A}$, and $\delta>0$, there exists $p'_0 \in \mathbb{A}$: $|p'_0-p_0|<\delta$, and an iterate $k>0$ such that
    \begin{align}
        |f^k(p'_0)-f^k(p_0)|\geqslant r.
    \end{align}\label{def_sensitive}
\end{definition}

Attractor $\mathbb{A}_L$ has already been taken as the smallest attracting set, so it is an indecomposable attractor by default.
We just need to prove that the system has sensitive dependence on initial conditions when restricted to $\mathbb{A}_L$.

We first prove that attractor $\mathbb{A}_P$ of the Poincar\'e map is chaotic using the fact that the doubling map \cite{robinson} is sensitively dependent upon initial conditions when restricted to its attractor.
Then we prove in exactly the same way that attractor $\mathbb{A}_L$ of the model system is chaotic by the sensitivity of $\mathbb{A}_P$.

\begin{proposition}[Sensitive Dependence of the Poincar\'e Map]
The Poincar\'e map of the model system has sensitive dependence upon initial conditions when restricted to its attractor $\mathbb{A}_P$.
\end{proposition}

\begin{proof}
For any $p_0\in \mathbb{A}_P$, with its neighboring initial point $p'_0\in \mathbb{A}_P$, we take $p'_0$ as $(x'_0,y'_0)=(x'_0,y'_0-\delta\cdot Sign(y_0))$.

Clearly, $\|p_n-p'_n\|\geqslant|y_n-y'_n|$.
So, proving sensitivity to initial conditions of $\mathbb{A}_P$ is equivalent to that of the Doubling Map:
\begin{align}
    y_{n+1}=\lfloor2y_n\rfloor,\quad (y_i\in(0,1),\forall i).
\end{align}
With the sensitive dependence upon initial conditions of doubling map when restricted to its attractor proved \cite{robinson},
Poincar\'e map of the model system is also sensitive when restricted to its attractor $\mathbb{A}_P$.
\end{proof}

In exactly the same way, the model system can thus be proved sensitively dependent upon initial conditions when restricted to its attractor $\mathbb{A}_L$.
Hence, it is a chaotic attractor by definition.


%
%
%

We also calculate the commonly used indicator of chaos: Lyapunov exponents \cite{robinson} for the model system at fixed points.
By solving the Lyapunov exponents in each coordinate direction, we find that in region $R_A$: $\ell_x=0$, $\ell_y=1$, and $\ell_z=-1$.
It can be found that there is a positive Lyapunov exponent $\ell_y=1$ denoting exponential expansion in the $y$ direction.
In region $R_B$ and $R_C$, $\ell_x=\ell_y=\ell_z=0$, which means that the expansion effect causing the sensitivity of the system is mainly exerted in region $R_A$.

From this and the former section, we find that the attractor of the model system is a chaotic attractor with fractal dimension: a strange chaotic attractor \cite{NonStrangeChaotic}.

%

\section{Construction of Potential Function in the Chaotic System}
Based on the above observation, we want to start constructing a potential function to describe the overall structure of the chaotic system.
First, we construct a ``seed function", $\mathcal{F}$, to account for the ``strangeness" of the system's attractor.
Then we prove its continuous differentiability so that it can be applied in the construction of potential function for the model system.
Later, we explicitly express the potential function in terms of the seed function $\mathcal{F}$.

\subsection{Definition of the ``Seed Function" $\mathcal{F}$}
\begin{definition}[Function $\mathcal{F}$]
\label{def:F}

Let
\begin{equation}
  f_1(x)=
  \left\{
    \begin{array}{l}
     0 \, ,\quad x\in\left[0\, , \;2/3\right]\bigcup\left[4/3\, , \;2\right]\\
     1-\cos(3\pi x) \, ,\quad x\in\left(2/3\, , \;4/3\right);
     \end{array}
    \right.\label{f1}
\end{equation}
and
\begin{equation}
  f_{n+1}(x)=
  \left\{
    \begin{array}{l}
     1/9\times f_n(3x) \, ,\quad x\in\left[0\, , \;2/3\right]\\
     0 \, ,\quad x\in\left(2/3\, , \;4/3\right)\\
     1/9\times f_n(3x-4) \, ,\quad x\in\left[4/3\, , \;2\right].\\
     \end{array}
    \right.\label{fn}
\end{equation}

Thus, we define the function $\mathcal{F}(x)$ as:
\begin{align}
    \mathcal{F}(x)=\sum_{n=1}^\infty f_n(x).
\end{align}
\end{definition}

Function $\mathcal{F}(x)$ defined on $[0,2]$ is shown in Figure (\ref{fig:mother_function}).
It has a fractal structure as the attractor $\mathbb{A}_P$ of the Poincar\'e map.

%

\begin{figure}
\begin{center}
  \includegraphics[width=0.45\textwidth]{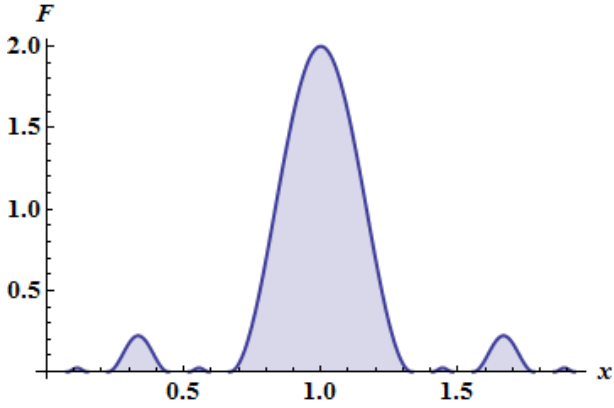}
\end{center}
\caption{
\begin{flushleft}
\textbf{Figure 5} $|$ \textbf{Illustration of the Self-Similar Function: $\mathcal{F}$.}
We hereby plot function $\mathcal{F}$ to intuitively visualize it.
With $\mathcal{F}$, construction of a potential function in the chaotic system would be natural.
\end{flushleft}
}
\label{fig:mother_function}
\end{figure}

\subsection{Proof of $\mathcal{F}(x)$ as Continuous Differentiable}

Here, we propose that the function $\mathcal{F}(x)$ defined above is continuously differentiable.

\begin{proposition}[Continuous Differentiability of Function $\mathcal{F}$]
\label{prop:F}
Function $\mathcal{F}(x)=\sum_{n=1}^\infty f_n(x)$ defined in Definition \ref{def:F} is continuously differentiable.
\end{proposition}

\begin{proof}
We first calculate the derivative of every $f_{n}(x)$ on $[0,2]$. Then we bound $f_{n}(x)$ and $|f'_{n}(x)|$ by geometric series to prove uniform convergence of their sums, and hence prove that $\mathcal{F}(x)=\sum_{n=1}^\infty f_n(x)$ is continuously differentiable.

From the definition (in equation (\ref{f1}) and (\ref{fn})):
\begin{equation}
  f_1(x)=
  \left\{
    \begin{array}{l}
     0 \, ,\quad x\in\left[0\, , \;2/3\right]\bigcup\left[4/3\, , \;2\right]\\
     1-\cos(3\pi x) \, ,\quad x\in\left(2/3\, , \;4/3\right),
     \end{array}
    \right.\nonumber
\end{equation}

and
\begin{equation}
  f_{n+1}(x)=
  \left\{
    \begin{array}{l}
     1/9\times f_n(3x) \, ,\quad x\in\left[0\, , \;2/3\right]\\
     0 \, ,\quad x\in\left(2/3\, , \;4/3\right)\\
     1/9\times f_n(3x-4) \, ,\quad x\in\left[4/3\, , \;2\right],\\
     \end{array}
    \right.\nonumber
\end{equation}

we have by taking derivative on both sides:
\begin{equation}
    f_1'(x)=
    \left\{\
    \begin{array}{l}
        0 \, , \quad x\in\left[0\, , \;2/3\right]\bigcup\left[4/3\, , \;2\right]\\
        3\pi\sin(3\pi x) \, , \quad x\in\left(2/3\, , \;4/3\right),
    \end{array}
    \right.\nonumber
\end{equation}

and
\begin{equation}
  f'_{n+1}(x)=
  \left\{
    \begin{array}{l}
     1/3\times f'_n(3x) \, ,\quad x\in\left[0\, , \;2/3\right]\\
     0 \, ,\quad x\in\left(2/3\, , \;4/3\right)\\
     1/3\times f'_n(3x-4) \, ,\quad x\in\left[4/3\, , \;2\right].\\
     \end{array}
    \right.\nonumber
\end{equation}

Obviously, $f_1(x)\leqslant 2$ and $|f'_1(x)|\leqslant 3\pi$. And also, $f'_n(x)$ is continuous for any $x\in[0,2]$.

If we further denote the set in which $f_n(x)$ is nonzero as $\mathbb{C}_n$,
we have:
\begin{align}
    \mathbb{C}_1=\left(\frac{2}{3},\frac{4}{3}\right),\nonumber
\end{align}
and
\begin{align}
    \mathbb{C}_{n+1}=\frac{\mathbb{C}_n}{3}\bigcup\frac{\mathbb{C}_n+4}{3}.
\end{align}
Since $\mathbb{C}_{n}\subset[0,2]$,
\begin{align}
    \frac{\mathbb{C}_n}{3}\bigcap\frac{\mathbb{C}_n+4}{3} \subset\left[0,\frac{2}{3}\right]\bigcap\left[\frac{4}{3},2\right]=\emptyset.\nonumber
\end{align}

We thus can conclude that:
\begin{align}
    f_{n+1}(x)\leqslant \frac{1}{9}f_{n}(x)\leqslant9^{-n}f_1(x)\leqslant2 \times 9^{-n},\nonumber
\end{align}
and
\begin{align}
    |f'_{n+1}(x)|\leqslant \frac{1}{3}|f'_{n}(x)|\leqslant3^{-n}|f'_1(x)|\leqslant3\pi \times 3^{-n},\nonumber
\end{align}
for any $x\in[0,2]$.

At this point, we give an upper bound for the series $\sum_{n=1}^m f_n(x)$ and its derivative (Although the least upper bound is even smaller if we note that $\mathbb{C}_n\bigcap\mathbb{C}_{n+1}$ is actually empty, an upper bound is good enough):
\begin{align}
    \sum_{n=1}^m f_n(x)\leqslant\frac{9}{4}-\frac{1}{4} \times 9^{-m},\nonumber
\end{align}
and that
\begin{align}
    \sum_{n=1}^m |f'_n(x)|\leqslant\frac{9\pi}{2}-\frac{3\pi}{2} \times 3^{-m}.\nonumber
\end{align}

So, $\sum_{n=1}^\infty f'_n(x)$ is uniformly absolutely-convergent. Hence, $\sum_{n=1}^\infty f'_n(x)$ and $\sum_{n=1}^\infty f_n(x)$ are all uniformly convergent.

With every $f_n(x)$ continuously differentiable, $\mathcal{F}(x)=\sum_{n=1}^\infty f_n(x)$ is continuously differentiable:
\begin{align}
    \frac{d}{dx}\mathcal{F}(x)=\sum_{n=1}^\infty f'_n(x).
\end{align}
\end{proof}

Clearly, the points where $\mathcal{F}(x)=\mathcal{F}'(x)=0$ form a Cantor set $\mathbb{C}$ corresponding to the attractor $\mathbb{A}_L$ of the model system ($\mathcal{F}(x)=0$ if and only if $x\in\mathbb{C}$).
At this point, we found that the potential function can be constructed in the following section.

\subsection{Constructing Potential Function in the Chaotic System}
As in equation (\ref{regionB}) we use $\theta$ to denote the angle that $(y,z)$ form with respect to the center $(2,2)$:
\begin{align}
    \theta=\arccos{\frac{y-2}{\sqrt{(z-2)^2+(y-2)^2}}}.\nonumber
\end{align}
Then, we construct potential function in each region respectively:

\begin{enumerate}[1.]
\item In the right part of region $R_A$, where $x,y,z\in[0,2]$, $yz\in[-2,2]$:
\begin{equation}
    \Psi_{A}=\left(\frac{4}{9\pi}\theta+\frac{5}{9}\right)\mathcal{F}(x).\label{PhiA}
\end{equation}

\item In region $R_B$, where $x\in[0,2]$, $y\in(2,4]$, $z\in[0,2]$, $\sqrt{(z-2)^2+(y-2)^2}\in[1,2]$:
\begin{equation}
    \Psi_{B}=\left(\frac{4}{9\pi}\theta+\frac{5}{9}\right) \mathcal{F}\left(\left(\frac{4}{3\pi}\theta+\frac{1}{3}\right)^{-1} x\right).\label{PhiB}
\end{equation}

\item In region $R_C$, where $x\in[0,2/3]$, $y\in[-2,4]$, $z>2$, $\sqrt{(z-2)^2+(y-2)^2}\geqslant 1$, $\sqrt{(z-2)^2+(y-3/2)^2}\leqslant 5/2$:
    \begin{equation}
    \Psi_{C}=\left(-\frac{4}{9\pi}\theta+\frac{5}{9}\right) \mathcal{F}(3x).\label{PhiC}
\end{equation}
\end{enumerate}

Here, we plot the potential function taken on the Poincar\'e section in Figure (\ref{fig:PotentialPoincare}).

\begin{figure}
\begin{center}
  \includegraphics[width=0.45\textwidth]{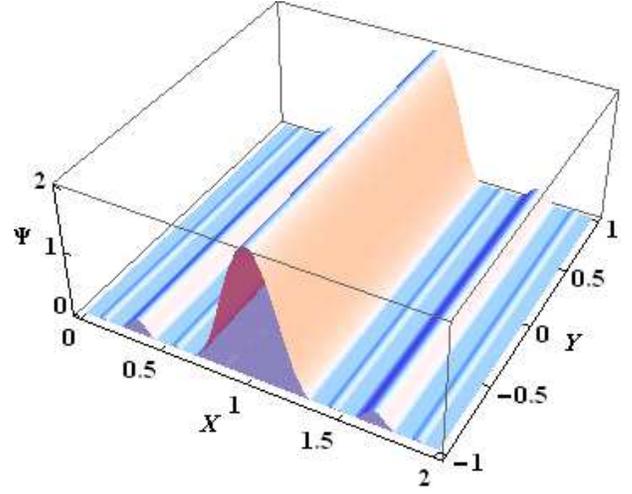}
\end{center}
\caption{
\begin{flushleft}
\textbf{Figure 6} $|$ \textbf{Potential Function on Poincar\'e Section.}\\
On the Poincar\'e section, the potential function is demonstrated to be a fractal object: it is zero when point $(x,y,2)$ belongs to the attractor $\mathbb{A}_L$, i.e., $\Psi|_{z=2}=0$ when $x \in \mathbb{C}$.
And when $(x,y,2)$ does not belong to the attractor, the potential function $\Psi$ has self-similar structure.
\end{flushleft}
}
\label{fig:PotentialPoincare}
\end{figure}

If we are also interested in the dynamics near saddle-focus fixed points, potential function in region $R_{D^A}$, $R_{D^B}$, and $R_{D^C}$ can also be constructed as follows.
\begin{enumerate}[1.]
\item In region $R_{D^A}$, where $x\in[0,2]$, $y\in[1,2]$, $z\in[1,2]$, $yz\in[2,4]$:
\begin{equation}
    \Psi_{D^A}=\left(\frac{4}{9\pi}\theta+\frac{5}{9}\right)\mathcal{F}(x)+1-\left(\dfrac{1}{2}yz-2\right)^2.\label{PhiDA}
\end{equation}

\item In region $R_{D^B}$, where $x\in[0 \, , \; 2/3+8/(3\pi)\times \theta]$, $y\in[2,3]$, $z\in[1,2]$, $\sqrt{(z-2)^2+(y-2)^2}\in[0,1]$:
\begin{align}
    \Psi_{D^B}=&\left(\frac{4}{9\pi}\theta+\frac{5}{9}\right) \mathcal{F}\left(\left(\frac{4}{3\pi}\theta+\frac{1}{3}\right)^{-1} x\right)\\\nonumber
    &+1-(z-2)^2-(y-2)^2.\label{PhiDB}
\end{align}

\item In region $R_{D^A}$, where $x\in[0, \, 2/3]$, $y\in[1,3]$, $z>2$, $\sqrt{(z-2)^2+(y-2)^2}\leqslant 1$:
    \begin{align}
    \Psi_{D^C}=&\left(-\frac{4}{9\pi}\theta+\frac{5}{9}\right) \mathcal{F}(3x)\\\nonumber
    &+1-(z-2)^2-(y-2)^2.\label{PhiDC}
    \end{align}
\end{enumerate}
Potential function in region $R_D$ is gradually higher in the center of the region than in its boundary with other regions.
Hence, points in this region will naturally converge to its boundary (at: $x=0, \, \sqrt{(z-2)^2+(y-2)^2}=1$, when $y>2$ or $z>2$; and at: $x=0, \, yz=2$, when $y, z \in [1,2]$).

The construction of the potential function is not unique.
If we change the expression of $f_1(x)$ in the definition of the ``seed function" $\mathcal{F}(x)$, we can have a different potential function for the system.

\section{Verification of the potential function}
Here, in this section, we verify the integrity of the potential function from three angles:
First, we show that it is continuous in the domain.
Second, we demonstrate that it decreases monotonically along the vector field.
Third, we show that $\nabla\Psi(x)=0$ if and only if $x$ belongs to the attractor of the system.

\subsection{Continuity of the Potential Function}

\begin{enumerate}[1.]
\item
At the boundary of region $R_A$ and $R_B$, $y=2$, $x\in[0,2]$, $z\in[0,1]$.

Hence, in equation (\ref{PhiA}), $\theta=\arccos{0}=\pi/2$;
$\left(4/(3\pi) \times \theta+1/3\right)^{-1}=1$;
and $\left(4/(9\pi) \times \theta+5/9\right)=7/9$.
So,
\begin{align}
    \Psi_A|_{y=2^-}=\frac{7}{9}\mathcal{F}(x)=\Psi_{B_1}|_{y=2^+}.\nonumber
\end{align}

\item
At the boundary of region $R_B$ and $R_C$, $z=2$, $x\in[0,2/3]$, and $y\in[3,4]$.

Hence, in equation (\ref{PhiB}), $\theta=\arccos{1}=0$;
$\left(4/(3\pi) \times \theta+1/3\right)^{-1}=3$;
and $\left(\pm4/(9\pi) \times \theta+5/9\right)=5/9$.
So,
\begin{align}
    \Psi_{B}|_{z=2^-}=\frac{5}{9}\mathcal{F}(3x)=\Psi_{C_1}|_{z=2^+}.\nonumber
\end{align}

\item
At the boundary of region $R_C$ and $R_A$, $z=2$, $x\in[0,2/3]$, and $y\in[-1,1]$.

Hence, in equation (\ref{PhiC}), $\theta=\arccos(-1)=\pi$;
$\left(4/(9\pi) \times \theta+5/9\right)=1$;
and $\left(-4/(9\pi) \times \theta+5/9\right)=1/9$.
So,
\begin{align}
    &\Psi_{A}|_{z=2^-}=\mathcal{F}(x)\nonumber\\
    &\Psi_{C}|_{z=2^+}=\frac{1}{9}\mathcal{F}(3x).\nonumber
\end{align}

It can be observed that in the definition of $\mathcal{F}(x)=\sum_{n=1}^\infty f_n(x)$,
\begin{align}
    f_{n+1}(x)=\frac{1}{9} f_n(3x) \, ,\quad x\in\left[0\, , \;2/3\right].\nonumber
\end{align}
%

Clearly, because $f_1(x)=0$ in $[0,2/3]$,
\begin{align}
    \frac{1}{9}\mathcal{F}(3x)=\sum_{n=2}^\infty f_n(x)=\mathcal{F}(x),
\end{align}
for $x\in[0,2/3]$.

\end{enumerate}

Please note that this is a critical point in the construction of potential function in cases of self-similar attractors. If the potential function constructed is not self-similar accordingly, then the boundary would not totally fit.

\subsection{Monotonic Decreasing of the Potential Function}
Then we take Lie derivative (derivative along the vector field) \cite{Arnold83Geo} of the potential function along vector fields in each region
remembering that $\mathcal{F}(x)\geqslant0$ for any $x\in[0,2]$.
\begin{enumerate}[1.]
\item In the right part of region $R_A$, where $x,y,z\in[0,2]$, $yz\in[-2,2]$:
\begin{equation}
\left\{
    \begin{array}{l}
     \dot x=0\\
     \dot \theta=Sgn(z-2)\\
     \left(-\dfrac{z-2}{(z-2)^2+(y-2)^2}  \dot{y}
     +\dfrac{y-2}{(z-2)^2+(y-2)^2}  \dot{z}\right)\\
     =-\dfrac{(2-y)z+(2-z)y}{(z-2)^2+(y-2)^2}.
     \end{array}
\right.\nonumber
\end{equation}
Since $\dot x=0$ and $y,z\in[0,2]$,
\begin{align}
    \dot\Psi_A&=\frac{4}{9\pi}\mathcal{F}(x)  \dot\theta\nonumber\\
    &=-\frac{4}{9\pi}\frac{(2-y)z+(2-z)y}{(z-2)^2+(y-2)^2} \mathcal{F}(x)\leqslant0.
\end{align}

\item In region $R_B$, where $x\in[0,2]$, $y\in(2,4]$, $z\in[0,2]$, $\sqrt{(z-2)^2+(y-2)^2}\in[1,2]$:
\begin{equation}
\left\{
    \begin{array}{l}
     \dot x=-x \left(\dfrac{\pi}{4}+\arccos{\dfrac{y-2}{\sqrt{(z-2)^2+(y-2)^2}}}\right)^{-1}\\
     =-x \left(\dfrac{\pi}{4}+\theta\right)^{-1}\\
     \dot\theta=\dfrac{z-2}{(z-2)^2+(y-2)^2}  \dot{y}
     -\dfrac{y-2}{(z-2)^2+(y-2)^2}  \dot{z}\\
     =-1.
     \end{array}
\right.\nonumber
\end{equation}

\begin{align}
    \dot\Psi_{B}&=\frac{4}{9\pi}
    \mathcal{F}\left(\left(\frac{4}{3\pi}\theta+\frac{1}{3}\right)^{-1}  x\right) \dot\theta\nonumber\\
    &+\left(\dfrac{4}{9\pi}\theta+\dfrac{5}{9}\right)\left(\dfrac{4}{3\pi}\theta+\dfrac{1}{3}\right)^{-2}
    \mathcal{F}'\left(\left(\frac{4}{3\pi}\theta+\frac{1}{3}\right)^{-1}  x\right)\nonumber\\
    &\left(\left(\frac{4}{3\pi}\theta+\frac{1}{3}\right) \dot x -\frac{4}{3\pi}x \dot\theta\right)\nonumber\\
    &=-\frac{4}{9\pi}
    \mathcal{F}\left(\left(\frac{4}{3\pi}\theta+\frac{1}{3}\right)^{-1}  x\right)\nonumber\\
    &+\left(\dfrac{4}{9\pi}\theta+\dfrac{5}{9}\right)\left(\dfrac{4}{3\pi}\theta+\dfrac{1}{3}\right)^{-2}
    \mathcal{F}'\left(\left(\frac{4}{3\pi}\theta+\frac{1}{3}\right)^{-1}  x\right)\nonumber\\
    &\left(-\left(\frac{4}{3\pi}\theta+\frac{1}{3}\right) \left(\frac{\pi}{4}+\theta\right)^{-1}  x +\frac{4}{3\pi}x\right).\nonumber
\end{align}

Since $\left(-\left(\dfrac{4}{3\pi}\theta+\dfrac{1}{3}\right) \left(\dfrac{\pi}{4}+\theta\right)^{-1}  x +\dfrac{4}{3\pi}x\right)=0$,
\begin{align}
    \dot\Psi_{B}=-\frac{4}{9\pi}\mathcal{F}\left(\left(\frac{4}{3\pi}\theta+\frac{1}{3}\right)^{-1}  x\right)\leqslant0.
\end{align}

\item In region $R_C$, where $x\in[0,2/3]$, $y\in[-1,4]$, $z>2$, $\sqrt{(z-2)^2+(y-2)^2}\geqslant 1$, $\sqrt{(z-2)^2+(y-3/2)^2}\leqslant 5/2$:
    \begin{equation}
    \left\{
        \begin{array}{l}
         \dot x=0\\
         \dot\theta=-\dfrac{z-2}{(z-2)^2+(y-2)^2} \dot{y}
         +\dfrac{y-2}{(z-2)^2+(y-2)^2} \dot{z}\\
         =\dfrac{(z-2)^2}{(z-2)^2+(y-2)^2}+\dfrac{y-2}{(z-2)^2+(y-2)^2}\\
         \left(\dfrac{9y}{8}- \dfrac{21}{8}+\dfrac{\sqrt{(3y-7)^2+8(z-2)^2}}{8}\right)>0.
        \end{array}
    \right.\nonumber
\end{equation}

Hence,
\begin{equation}
    \dot\Psi_{C}=-\frac{4}{9\pi} \mathcal{F}(3x) \dot\theta \leqslant0.
\end{equation}
\end{enumerate}

\subsection{Potential Function and the Attractor}
Here, we verify that the potential function attains extremum: $\nabla\Psi(\mathbf{x})=0$ if and only if $\mathbf{x}=(x, y, z)$ belongs to the attractor $\mathbb{A}_L$ of the system.
Again, $\theta=\arccos{(y-2)/\sqrt{(z-2)^2+(y-2)^2}}$.

\begin{enumerate}[1.]
\item In the right part of region $R_A$, where $x,y,z\in[0,2]$, $yz\in[-2,2]$:
\begin{equation}
    \nabla\Psi_A=
    \left(
        \begin{array}{l}
         \left(\dfrac{4}{9\pi}\theta+\dfrac{5}{9}\right)\mathcal{F}'(x)\\
         \dfrac{4}{9\pi} \dfrac{z-2}{(z-2)^2+(y-2)^2}\mathcal{F}(x)\\
         -\dfrac{4}{9\pi} \dfrac{y-2}{(z-2)^2+(y-2)^2}\mathcal{F}(x)
        \end{array}
    \right).\nonumber
\end{equation}

Here, $\mathcal{F}(x)=0$ and $\mathcal{F}'(x)=0$ if and only if $x\in\mathbb{C}$;
and point $(x, y, z)$ belongs to the attractor $\mathbb{A}_L$ if and only if $x\in\mathbb{C}$.
So, $\nabla\Psi_A=0$ if and only if $(x, y, z)\in \mathbb{A}_L$.

\begin{widetext}
\item In region $R_B$, where $x\in[0,2]$, $y\in(2,4]$, $z\in[0,2]$, $\sqrt{(z-2)^2+(y-2)^2}\in[1,2]$:
%
%
%
%
%
%

\begin{equation}
    \nabla\Psi_{B}=
    \left(
        \begin{array}{l}
         \left(\dfrac{4}{9\pi}\theta+\dfrac{5}{9}\right)\left(\dfrac{4}{3\pi}\theta+\dfrac{1}{3}\right)^{-1}\mathcal{F}'\left(\left(\dfrac{4}{3\pi}\theta+\dfrac{1}{3}\right)^{-1} x\right)\\

         \dfrac{4}{9\pi} \dfrac{z-2}{(z-2)^2+(y-2)^2}
         \left(\mathcal{F}\left(\left(\dfrac{4}{3\pi}\theta+\dfrac{1}{3}\right)^{-1} x\right)

-\left(\dfrac{4}{3\pi}\theta+\dfrac{5}{3}\right)\left(\dfrac{4}{3\pi}\theta+
        \dfrac{1}{3}\right)^{-2}\mathcal{F}'\left(\left(\dfrac{4}{3\pi}\theta+\dfrac{1}{3}\right)^{-1} x\right)\right)\\

         -\dfrac{4}{9\pi} \dfrac{y-2}{(z-2)^2+(y-2)^2}
         \left(\mathcal{F}\left(\left(\dfrac{4}{3\pi}\theta+\dfrac{1}{3}\right)^{-1} x\right)

-\left(\dfrac{4}{3\pi}\theta+\dfrac{5}{3}\right)\left(\dfrac{4}{3\pi}\theta+
\dfrac{1}{3}\right)^{-2}\mathcal{F}'\left(\left(\dfrac{4}{3\pi}\theta+\dfrac{1}{3}\right)^{-1} x\right)\right)
        \end{array}
    \right).\nonumber
\end{equation}

Here, $\mathcal{F}\left(\left(4/(3\pi)\times\theta+1/3\right)^{-1}x\right)=0$
and $\mathcal{F}'\left(\left(4/(3\pi)\times\theta+1/3\right)^{-1}x\right)=0$
if and only if $\left(4/(3\pi)\times\theta+1/3\right)^{-1}x\in\mathbb{C}$;
and point $(x, y, z)$ belongs to the attractor $\mathbb{A}_L$ if and only if
$\left(4/(3\pi)\times\theta+1/3\right)^{-1}x\in\mathbb{C}$.
So, $\nabla\Psi_{B}=0$ if and only if $(x, y, z)\in \mathbb{A}_L$.
\end{widetext}

\item In region $C$, where $x\in[0,2/3]$, $y\in[-1,4]$, $z>2$, $\sqrt{(z-2)^2+(y-2)^2}\geqslant 1$, $\sqrt{(z-2)^2+(y-3/2)^2}\leqslant 5/2$:
    \begin{equation}
    \nabla\Psi_{C}=
    \left(
        \begin{array}{l}
         \left(-\dfrac{4}{3\pi}\theta+\dfrac{5}{3}\right)\mathcal{F}'(3x)\\
         \dfrac{4}{9\pi} \dfrac{z-2}{(z-2)^2+(y-2)^2}\mathcal{F}(3x)\\
         -\dfrac{4}{9\pi} \dfrac{y-2}{(z-2)^2+(y-2)^2}\mathcal{F}(3x)
        \end{array}
    \right).\nonumber
\end{equation}

Here, $\mathcal{F}(3x)=0$ and $\mathcal{F}'(3x)=0$ if and only if $x\in\mathbb{C}$;
and point $(x, y, z)$ belongs to the attractor $\mathbb{A}_L$ if and only if $x\in\mathbb{C}$.
So, $\nabla\Psi_{C}=0$ if and only if $(x, y, z)\in \mathbb{A}_L$.

\end{enumerate}

In exactly the same way, can we also show the integrity of the potential function in region $R_D$.

\section{Comparison with Related Works}
As discussed in the introduction, constructing a potential-like function in the chaotic system is actually an effort that is by no means totally strange to researchers.
Until recently, there are various efforts seeking to describe chaotic dynamics using generalized Hamiltonian approach \cite{Sira}, energy-like function technique \cite{Sarasola}, minimum action method \cite{Weinan}, and etc.
These previous methods all construct a potential-like scaler function to analyze certain chaotic system.
Unfortunately, the scalar functions in these works all lack certain important properties.

For example, the generalized Hamiltonian systems approach takes a quadratic form of the state variables as the ``generalized Hamiltonian" \cite{Sira}, a Hamiltonian that includes conserved dynamics, energy dissipation, and energy input.
The third part transforms an autonomous differential equation into a non-autonomous physical model, attempting to explain for the ``irregular" \cite{Lorenz} behavior of chaotic systems.
Contrary to this expectation,
when the Hamiltonian is set possible to dissipate and increase, the generalized Hamiltonian itself becomes a chaotic oscillating signal with respect to time.
Therefore, it remains an issue as to what additional insight this generalized Hamiltonian can provide about the original system, such as global stability or local performance.

The energy-like function technique is essentially similar to the generalized Hamiltonian approach.
Its energy-like function differs from the generalized Hamiltonian in a way that it may not be a quadratic form of the state variables.
Rather, the energy-like function is constructed based on the ``geometric appearance" \cite{Sarasola} of the attractor corresponding to the specific chaotic system.
Although this technique would seem more sophisticated, its energy-like function still oscillate chaotically along with time, describing chaotic dynamics in a chaotic fashion.
Loss of monotonicity restricts the function from describing the system's essential properties like stability and performance.

The minimum action method, however, cast the problem under the light of zero noise limit.
By constructing an auxiliary Hamiltonian \cite{WentzellRecent} (commonly denoted as ``Freidlin-Wentzell Hamiltonian"), Freidlin-Wentzell action functional can be minimized \cite{Wentzell}.
This method analyzes chaotic system by possible transitions between limit sets \cite{Weinan}.
But since the Freidlin-Wentzell Hamiltonian can be not bounded even in globally stable systems, it is not a quantitative measure comparable between points in state space, hence, not an ideal potential function.

In short, all the previous works each focuses on one attribute of the potential function.
However, as we can see from our constructive result, only when all the requirements (in definition $1$) are met, would the potential function reflect evolution of the whole system and structure of the chaotic attractor.
In this sense, the current work is the first construction to satisfy such strong conditions, providing a both detailed and global description for a chaotic system.

\section{Chaotic attractor and strange attractor}

\begin{figure}
\begin{center}
    \includegraphics[width=0.5\textwidth]{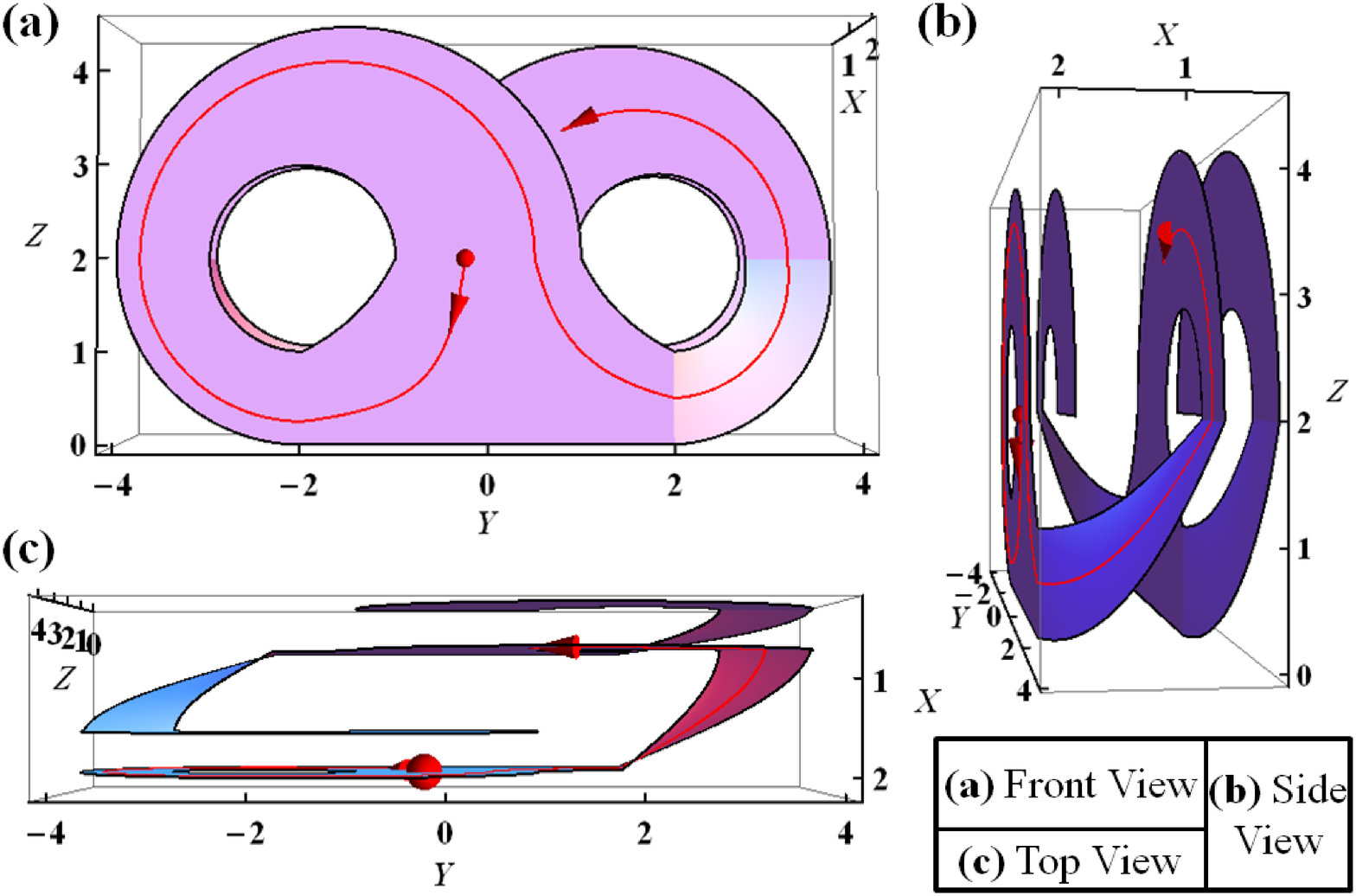}
\end{center}
\caption{
\begin{flushleft}
\textbf{Figure 7} $|$ \textbf{Strange Chaotic Attractor.}\\
We find that a connected surface of the attractor is of infinite layers.
We show how surface $x=2$ is linked to the other equipotential layers.
It can be seen that the surface of the attractor is orientable.
We hereby demonstrate the strange chaotic attractor viewed from (a), front; (b), side; and (c), top.
The trajectory running from point (2, 1/4, 2) is also shown in the figure.
\end{flushleft}
}
\label{fig:Strange}
\end{figure}
With the potential function constructed, we can easily solve the system's attractor without any need of numerical simulation.
We find that the attractor is composed of connected surfaces, each of infinite layers.
Starting from the plane $x=2$, we show the configuration of these layers in Figure (\ref{fig:Strange}).
Since geometric configuration of chaotic attractor interests many researchers \cite{Gilmore}, we demonstrate in the figure that the chaotic attractor of the system studied in this paper consists of orientable surfaces.

The chaotic attractor here is a strange attractor of fractal dimension \cite{NonStrangeChaotic}.
And in the literature of dynamical systems, there have long been discussions about the relationship between chaotic attractors and strange attractors \cite{StrangeAndChaotic}.
Several examples of strange nonchaotic attractors and nonstrange chaotic attractors have been found \cite{NonStrangeChaotic}.
Until recently, strange nonchaotic attractors are still studied \cite{StrangeNonChaotic}.

The potential function approach provides a unified framework to treat the topics of chaotic attractors and strange attractors together.
To clarify this insight, we need to apply our decomposition method (equation (\ref{decomp})) here:
\begin{align}
\dot{\mathbf{x}}=\mathbf{f}(\mathbf{x})=-D \nabla\Psi(\mathbf{x})+Q \nabla\Psi(\mathbf{x}).\nonumber
\end{align}
where
\begin{align}
    D=-\frac{\mathbf{f}\cdot\nabla\Psi}{\nabla\Psi\cdot\nabla\Psi}I,\nonumber
\end{align}
and
\begin{align}
    Q=\frac{\mathbf{f}\times\nabla\Psi}{\nabla\Psi\cdot\nabla\Psi}.\nonumber
\end{align}

We first analyze our model system with this decomposition framework.
Then we further modify our model system to two typical cases interesting to many researchers: a nonstrange chaotic attractor and a strange nonchaotic attractor.
After analyzing these two cases, we explain the different origins of chaotic attractors and strange attractors in general.

\subsection{Decomposition of the Chaotic System}
According to our decomposition scheme, we first decompose the chaotic dynamical system in each region into two components: the gradient component and the rotation component.

In region $R_A$, $\nabla\Psi_A$ is solved as:
\begin{equation}
    \nabla\Psi_A=
    \left(
        \begin{array}{l}
         \left(\dfrac{4}{9\pi}\theta+\dfrac{5}{9}\right)\mathcal{F}'(x)\\
         -\dfrac{4}{9\pi} \dfrac{2-z}{(z-2)^2+(y-2)^2}\mathcal{F}(x)\\
         \dfrac{4}{9\pi} \dfrac{2-y}{(z-2)^2+(y-2)^2}\mathcal{F}(x)
        \end{array}
    \right).\nonumber
\end{equation}

Hence, we can find the expression of the matrix $D_A$ accounting for the gradient component of the vector field in region $R_A$:

\begin{align}
    D_A&=-\frac{\mathbf{f}_A\cdot\nabla\Psi_A}{\nabla\Psi_A\cdot\nabla\Psi_A}I \\\nonumber
       &=\dfrac{\dfrac{4}{9\pi} \dfrac{(2-y)z+(2-z)y}{(z-2)^2+(y-2)^2}\mathcal{F}(x)}
       {\left(\dfrac{4}{9\pi}\theta+\dfrac{5}{9}\right)^2\left(\mathcal{F}'(x)\right)^2
       +\dfrac{\left(4/9\pi\right)^2}{(z-2)^2+(y-2)^2}\left(\mathcal{F}(x)\right)^2}I.
\end{align}

\begin{widetext}
The decomposed gradient part would then be:
\begin{equation}
    D_A\nabla\Psi_A=
    \dfrac{\dfrac{4}{9\pi} \dfrac{(2-y)z+(2-z)y}{(z-2)^2+(y-2)^2}\mathcal{F}(x)}
       {\left(\dfrac{4}{9\pi}\theta+\dfrac{5}{9}\right)^2\left(\mathcal{F}'(x)\right)^2
       +\dfrac{\left(4/9\pi\right)^2}{(z-2)^2+(y-2)^2}\left(\mathcal{F}(x)\right)^2}\\
    \left(
        \begin{array}{l}
         \left(\dfrac{4}{9\pi}\theta+\dfrac{5}{9}\right)\mathcal{F}'(x)\\
         -\dfrac{4}{9\pi} \dfrac{2-z}{(z-2)^2+(y-2)^2}\mathcal{F}(x)\\
         \dfrac{4}{9\pi} \dfrac{2-y}{(z-2)^2+(y-2)^2}\mathcal{F}(x)
        \end{array}
    \right).\nonumber
\end{equation}

Also, we can find the decomposed rotation part by finding $Q_A$ as:
\begin{align}
    Q_A&=\frac{\mathbf{f}_A\times\nabla\Psi_A}{\nabla\Psi_A\cdot\nabla\Psi_A} \\\nonumber
       &=\dfrac{1}{{\left(\dfrac{4}{9\pi}\theta+\dfrac{5}{9}\right)^2\left(\mathcal{F}'(x)\right)^2
       +\dfrac{\left(4/9\pi\right)^2}{(z-2)^2+(y-2)^2}\left(\mathcal{F}(x)\right)^2}} \\\nonumber
&\left(
      \begin{array}{ccc}
        0&      -y\left(\dfrac{4}{9\pi}\theta+\dfrac{5}{9}\right)\mathcal{F}'(x)&      z\left(\dfrac{4}{9\pi}\theta+\dfrac{5}{9}\right)\mathcal{F}'(x) \\
        y\left(\dfrac{4}{9\pi}\theta+\dfrac{5}{9}\right)\mathcal{F}'(x)&       0&      -\dfrac{4}{9\pi} \dfrac{(2-z)z-(2-y)y}{(z-2)^2+(y-2)^2}\mathcal{F}(x) \\
        -z\left(\dfrac{4}{9\pi}\theta+\dfrac{5}{9}\right)\mathcal{F}'(x)&       \dfrac{4}{9\pi} \dfrac{(2-z)z-(2-y)y}{(z-2)^2+(y-2)^2}\mathcal{F}(x)&         0
      \end{array}
\right).
\end{align}

The decomposed rotation part would then be:
\begin{align}
    Q_A\nabla\Psi_A=
    &\dfrac{1}
       {\left(\dfrac{4}{9\pi}\theta+\dfrac{5}{9}\right)^2\left(\mathcal{F}'(x)\right)^2
       +\dfrac{\left(4/9\pi\right)^2}{(z-2)^2+(y-2)^2}\left(\mathcal{F}(x)\right)^2} \\\nonumber
    &\left(
        \begin{array}{l}
         -\dfrac{4}{9\pi}\left(\dfrac{4}{9\pi}\theta+\dfrac{5}{9}\right)\dfrac{(y-2)z+(z-2)y}{(z-2)^2+(y-2)^2}\mathcal{F}'(x)\mathcal{F}(x)\\
         y\left(\dfrac{4}{9\pi}\theta+\dfrac{5}{9}\right)^2\left(\mathcal{F}'(x)\right)^2+ \left(\dfrac{4}{9\pi}\right)^2\dfrac{(y-2)^2y-(y-2)(z-2)z}{\left((z-2)^2+(y-2)^2\right)^2}\left(\mathcal{F}(x)\right)^2\\
         -z\left(\dfrac{4}{9\pi}\theta+\dfrac{5}{9}\right)^2\left(\mathcal{F}'(x)\right)^2+ \left(\dfrac{4}{9\pi}\right)^2\dfrac{(y-2)y(z-2)-(z-2)^2z}{\left((z-2)^2+(y-2)^2\right)^2}\left(\mathcal{F}(x)\right)^2
        \end{array}
    \right).\nonumber
\end{align}

\end{widetext}

When the system approaches its attractor, i.e., $\mathcal{F}(x)\rightarrow0$,
\begin{align}
    \dfrac{\mathcal{F}(x)}{\left(\mathcal{F}'(x)\right)^2}
    &=\lim_{x\rightarrow 0}\dfrac{\left(\dfrac{1}{9}\right)^n \left(1-\cos(3\pi x)\right)} {\left(\left(\dfrac{1}{3}\right)^n 3\pi \sin(3\pi x)\right)^2}\\\nonumber
    &=\lim_{x\rightarrow 0}\dfrac{\left(\dfrac{1}{9}\right)^n \dfrac{1}{2}(3\pi x)^2}
    {\left(\left(\dfrac{1}{3}\right)^n 9\pi^2x\right)^2}
    =\dfrac{1}{18\pi^2}.
\end{align}

Thus, when the system converges to its attractor, the gradient matrix $D_A$ would be:
\begin{align}
    D_A&=-\dfrac{\dfrac{4}{9\pi} \dfrac{(y-2)z+(z-2)y}{(z-2)^2+(y-2)^2}}
    {\left(\dfrac{4}{9\pi}\theta+\dfrac{5}{9}\right)^2}
    \frac{\mathcal{F}(x)}{\left(\mathcal{F}'(x)\right)^2}I\\\nonumber
    &=-\dfrac{2}{81\pi^3} \dfrac{\dfrac{(y-2)z+(z-2)y}{(z-2)^2+(y-2)^2}}
    {\left(\dfrac{4}{9\pi}\theta+\dfrac{5}{9}\right)^2}I,
\end{align}
which is finite.

Hence, the gradient component $D_A\nabla\Psi_A$ of the system would converges to zero when approaching the attractor.
So the motion on the attractor is caused totally by the rotation part: $Q_A\nabla\Psi_A$.

In exactly the same way, the decomposition procedure can be carried out in region B and region C, and the same conclusion holds.

\subsection{Nonstrange Chaotic Attractor}
Let's first examine an example of nonstrange chaotic attractor by modifying our original system a little (in region $R_B$, equation (\ref{regionB})):

In region $R_B$, we set
\begin{align}
    \theta=\arccos{\dfrac{y-2}{\sqrt{(z-2)^2+(y-2)^2}}},\nonumber
\end{align}
as in equation (\ref{regionB}).
Then we change the dynamical system in region $R_B$ (defined as $x\in[0 \, , \; 4/\pi \times \theta]$, $y\in[2,4]$, $z\in[0,2]$, $\sqrt{(z-2)^2+(y-2)^2}\in[1,2]$) to:
\begin{equation}
\left\{
    \begin{array}{l}
     \dot x=-\dfrac{x}{\theta}\\
     \dot y=2-z\\
     \dot z=y-2.
     \end{array}
\right.\nonumber
\end{equation}
The same as in the original system, domain of definition can be expanded to the whole $\mathbb{R}^3$ space.

Consequently, the Poincar\'e map would be as follows:

When $(x,y)\in [0,2]\times[0,1]$,
\begin{equation}
\left\{
    \begin{array}{l}
     x_{n+1}=0\\
     y_{n+1}=2y_n-1.
     \end{array}
\right.\nonumber
\end{equation}

When $(x,y)\in [0,2]\times[-1,0)$,
\begin{equation}
\left\{
    \begin{array}{l}
     x_{n+1}=1\\
     y_{n+1}=2y_n+1.
     \end{array}
\right.\nonumber
\end{equation}

The attractor $\mathbb{A}_L'$ of the modified system would be:
(assuming $\theta=\arccos{(y-2)/\sqrt{(z-2)^2+(y-2)^2}}$):

In region $R_A$, $x=0$ or $2$;

in region $R_B$, $(\pi/2)\times(x/\theta)=0$ or $2$;

in region $R_C$, $x=0$ or $2$.

The attractor is shown in Figure (\ref{fig:Nonstrange}).
We can calculate its box-counting dimension to be:
\begin{align}
d_b(\mathbb{A}_L')=\lim_{\epsilon\rightarrow0}\dfrac{\log N(\epsilon, \mathbb{A}_L')}{\log(1/\epsilon)}=2,
\end{align}
which is an integer dimension.
Actually, the attractor $\mathbb{A}_L'$ is just two orientable surfaces folded together.
Hence, it is no longer a strange attractor anymore.

Exactly as in the original system, the modified attractor can be proved to be chaotic.
And we can also calculate the commonly used indicator of chaos: Lyapunov exponents \cite{robinson} for the model system at fixed points.
Lyapunov exponents are solved in each direction as: $\ell_x=0$, $\ell_y=1$, and $\ell_z=-1$ in region $R_A$ (In other regions, $\ell_x=\ell_y=\ell_z=0$).
It is found that there is a positive Lyapunov exponent $\ell_y=1$ denoting exponential expansion in the $y$ direction, exactly as in the original model system.

So, it is clear that the modified attractor is a nonstrange chaotic attractor.

Now, we construct a potential function $\Phi$ for the new dynamical system by first appointing a new seed function $F(x)$ defined in $[0,2]$:
\begin{equation}
F(x)=
     1-\cos(\pi x) \, ,\quad x\in\left[0\, , \;2\right].\nonumber
\end{equation}

Hence, the potential function can be represented as:

\begin{enumerate}[1.]
\item In the right part of region $R_A$, where $x,y,z\in[0,2]$:
\begin{equation}
    \Phi_A=\left(\frac{\theta}{\pi}\right)F(x).\nonumber
\end{equation}

\item In region $R_B$, where $y\in[2,4]$, $z\in[0,2]$, $\sqrt{(z-2)^2+(y-2)^2}\in[1,2]$, $x\in[0 \, , \; 4/\pi \times \theta]$:
\begin{equation}
    \Phi_{B}=\left(\frac{\theta}{\pi}\right) F\left(\frac{\pi x}{2 \theta}\right).\nonumber
\end{equation}

\item In region $R_C$, where $x=2$, $y\in[-2,4]$, $z>2$, $\sqrt{(z-2)^2+(y-2)^2}\geqslant 1$, $\sqrt{(z-2)^2+(y-3/2)^2}\leqslant 5/2$:
\begin{equation}
    \Phi_{C}=0.\nonumber
\end{equation}
\end{enumerate}

Here, $\Phi=0$ corresponds to the attractor.

We can further decompose the system as with the original model system:
\begin{align}
\dot{\mathbf{x}}=\mathbf{f}(\mathbf{x})=-D \nabla\Phi(\mathbf{x})+Q \nabla\Phi(\mathbf{x}).\nonumber
\end{align}

Then, in region $R_A$:
\begin{equation}
    \nabla\Phi_A=\dfrac{1}{\pi}
    \left(
        \begin{array}{l}
         \theta{F}'(x)\\
         -\dfrac{2-z}{(z-2)^2+(y-2)^2}{F}(x)\\
         \dfrac{2-y}{(z-2)^2+(y-2)^2}{F}(x)
        \end{array}
    \right).\nonumber
\end{equation}

\begin{align}
    D_A&=-\frac{\mathbf{f}_A\cdot\nabla\Phi_A}{\nabla\Phi_A\cdot\nabla\Phi_A}I \\\nonumber
       &=\dfrac{\pi \dfrac{(2-y)z+(2-z)y}{(z-2)^2+(y-2)^2}{F}(x)}
       {\theta^2\left({F}'(x)\right)^2
       +\dfrac{1}{{(z-2)^2+(y-2)^2}}\left({F}(x)\right)^2}I.
\end{align}

The decomposed gradient part would then be:
\begin{align}
    D_A\nabla\Phi_A=
    &\dfrac{\dfrac{(2-y)z+(2-z)y}{(z-2)^2+(y-2)^2}{F}(x)}
       {\theta^2\left({F}'(x)\right)^2
       +\dfrac{\left({F}(x)\right)^2}{{(z-2)^2+(y-2)^2}}}\\
    &\left(
        \begin{array}{l}
         \theta{F}'(x)\\
         -\dfrac{2-z}{(z-2)^2+(y-2)^2}{F}(x)\\
         \dfrac{2-y}{(z-2)^2+(y-2)^2}{F}(x)
        \end{array}
    \right).\nonumber
\end{align}

\begin{widetext}
Also, we can find the decomposed rotation part by finding $Q_A$ as:
\begin{align}
    Q_A&=\frac{\mathbf{f}_A\times\nabla\Phi_A}{\nabla\Phi_A\cdot\nabla\Phi_A} \\\nonumber
       &=\dfrac{\pi}{\theta^2\left({F}'(x)\right)^2
       +\left({F}(x)\right)^2}
&\left(
      \begin{array}{ccc}
        0&      -y\theta{F}'(x)&      z\theta{F}'(x) \\
        y\theta{F}'(x)&       0&      -\dfrac{(2-z)z-(2-y)y}{(z-2)^2+(y-2)^2}{F}(x) \\
        -z\theta{F}'(x)&       \dfrac{(2-z)z-(2-y)y}{(z-2)^2+(y-2)^2}{F}(x)&         0
      \end{array}
\right).
\end{align}

The decomposed rotation part would then be:
\begin{align}
    Q_A\nabla\Phi_A=
    &\dfrac{\pi}{\theta^2\left({F}'(x)\right)^2
       +\left({F}(x)\right)^2}\\\nonumber
    &\left(
        \begin{array}{l}
         -\dfrac{(y-2)z+(z-2)y}{(z-2)^2+(y-2)^2}\theta{F}'(x){F}(x)\\
         y\theta^2\left({F}'(x)\right)^2+ \dfrac{(y-2)^2y-(y-2)(z-2)z}{\left((z-2)^2+(y-2)^2\right)^2}\left({F}(x)\right)^2\\
         -z\theta^2\left({F}'(x)\right)^2+ \dfrac{(y-2)y(z-2)-(z-2)^2z}{\left((z-2)^2+(y-2)^2\right)^2}\left({F}(x)\right)^2
        \end{array}
    \right).\nonumber
\end{align}

\end{widetext}

The properties of the gradient part and the rotation part corresponds exactly to the original model system.
That is: the gradient part converges to zero when approaching the attractor;
motion on the attractor is determined by the rotation part.

\begin{figure}
\begin{center}
  \includegraphics[width=0.5\textwidth]{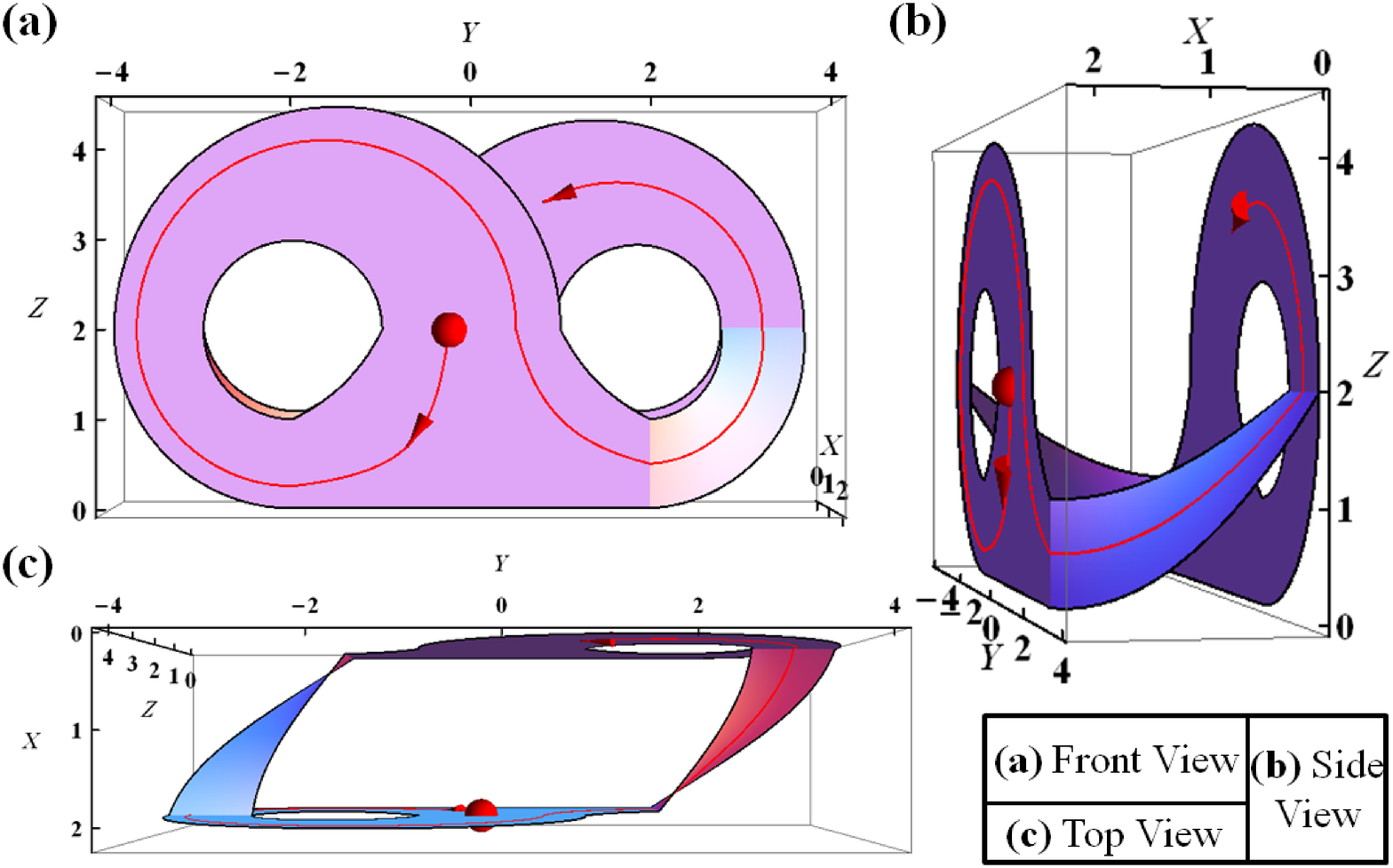}
\end{center}
\caption{
\begin{flushleft}
\textbf{Figure 8} $|$ \textbf{Nonstrange Chaotic Attractor.}\\
We change the expression of the system a little, so that the attractor is just two orientable surfaces folded together, rather than a fractal structure.
However, it remains to be a chaotic attractor.
We hereby demonstrate the nonstrange chaotic attractor viewed from (a), front; (b), side; and (c), top.
The trajectory running from point (2, 1/4, 2) is also shown in the figure.
\end{flushleft}
}
\label{fig:Nonstrange}
\end{figure}

\subsection{Strange Nonchaotic Attractor}
A strange nonchaotic attractor can also be constructed.

We simply take the gradient of potential function $\Psi$ of the original system in each region of definition:
\begin{equation}
\left\{
    \begin{array}{l}
     \dot x=-\partial_x\Psi\\
     \dot y=-\partial_y\Psi\\
     \dot z=-\partial_z\Psi.
     \end{array}
\right.\nonumber
\end{equation}

If we take left part of region $R_A$ for example, the vector field would be:
\begin{equation}
    \mathbf{f}_A=
    \left(
        \begin{array}{l}
         -\left(\dfrac{4}{9\pi}\theta+\dfrac{5}{9}\right)\mathcal{F}'(x)\\
         -\dfrac{4}{9\pi} \dfrac{z-2}{(z-2)^2+(y-2)^2}\mathcal{F}(x)\\
         \dfrac{4}{9\pi} \dfrac{y-2}{(z-2)^2+(y-2)^2}\mathcal{F}(x)
        \end{array}
    \right).\nonumber
\end{equation}

The resultant ODE system defined by the gradient is a dynamical since it is Lipschitz continuous in each region.
And with existence and uniqueness of the flow guaranteed by Lipschitz continuity, conditions for the system being a dynamical system can be satisfied and extended to include boundaries.

The system would converge downward the potential function $\Psi$ until reaching the states where $\Psi=0$.
Consequently, attractor of this system is characterized by $\Psi=0$, as in the original model system.
Hence, the gradient system's attractor is the same attractor $\mathbb{A}_L$ of the original model system, whose box-counting dimension:
\begin{align}
    d_b(\mathbb{A}_L)=\lim_{\epsilon\rightarrow0}\dfrac{\log N(\epsilon,\mathbb{A}_L)}{\log(1/\epsilon)}=2+ln(2)/ln(3).
\end{align}
Hence, the system has a strange attractor.

Since $\nabla\Psi=0$ when $\Psi=0$, the dynamical system is not sensitively dependent upon initial conditions when restricted to the attractor.
So, the attractor is not chaotic.
Also, its Lyapunov exponents at the fixed points (where $x\in\mathbb{C}$, $y=z=0$) would be: $\ell_x=-(1/3)^n\times 9\pi^2$, $\ell_y=\ell_z=0$.
Hence, it's a strange nonchaotic attractor.

Decomposition of this system would give: $D(\mathbf{x})=I$ and $Q(\mathbf{x})=0$.
Thus, $D\nabla\Psi(\mathbf{x})=-\mathbf{f}(\mathbf{x})$ is just the reversed gradient system.


\subsection{Chaotic Attractor versus Strange Attractor}
The previous two examples show that the concepts of chaotic attractor and strange attractor do not imply each other.
Under our framework of decomposition (equation (\ref{decomp})) here:
\begin{align}
\dot{\mathbf{x}}=\mathbf{f}(\mathbf{x})=-D \nabla\Psi(\mathbf{x})+Q \nabla\Psi(\mathbf{x}).\nonumber
\end{align}

Since $\dot\Psi=\nabla\Psi\cdot\dot{\mathbf{x}}=\nabla\Psi\cdot(D \nabla\Psi)$, $\Psi$ decreases monotonically according to the gradient component $D \nabla\Psi$ of the vector field $f$.
Then the attractor is naturally characterized by $D \nabla\Psi=0$.
So, whether the attractor is a strange attractor is determined by the gradient part of the vector field.

Sensitive dependence upon initial conditions when restricted to the attractor, however, is determined by the rotation part of the vector field: $Q \nabla\Psi$.
Once the system has evolved to the limit set, $D \nabla\Psi$ would equal to zero, and $Q \nabla\Psi$ would be prevalent.
Hence, the rotational vector field on the attractor causes the expansion of the state space, leading to dynamical sensitivity.
Conversely, when $Q \nabla\Psi=0$, chaotic motion on the attractor would not exist.
In this sense, nonzero rotation part of the dynamical system is a necessary condition for causing hyperbolic chaos \cite{Kuznetsov}.

To sum up, gradient part and rotation part of the vector field are responsible for the creation of strange attractor and chaotic attractor correspondingly.
Although they are both affected by the geometrical configuration of the potential function $\Psi$, they denote dissipation and circulation respectively.


\section{Conclusion}
In the present paper, it is shown that potential functions with monotonic properties can be constructed in continuous dissipative chaotic systems with strange attractors.
The potential function here is a continuous function in phase space, monotonically decreasing with time and remains constant if and only if limit set is reached.
This definition is a natural restriction of generic dynamics since it is a direct generalization of Lyapunov function and corresponds to the concept of energy.

Potential function defined this way also implies that the dynamics can be decomposed into two parts: a gradient part, dissipating energy potential; and a rotation part, conserving energy potential.
The gradient part drives the system towards the attractor while the rotation part perpetuates the system's circular motion on the attractor.

To demonstrate the power of this framework in chaotic systems, we simplify the geometric Lorenz attractor, and prove by definition that it is a chaotic attractor.
Then we analytically and explicitly construct a suitable potential function for the attractor, which, to our best knowledge, is the first example in chaotic dynamics.
The potential function reveals the fractal nature of the chaotic strange attractor.

We further analyze the concept of chaotic attractor and strange attractor with our decomposition.
It is found that chaotic attractor originates in the rotation part, prompting the state space of the attractor to expand; while strange attractor originates in the gradient part, causing initial states attracted to complex limit set.

\section*{Acknowledgements}
The authors would like to express their sincere gratitude
to Xinan Wang, Ying Tang, Song Xu and Jianghong Shi for their constructive advice throughout this work.
The authors also appreciate valuable discussion with James A. Yorke, David Cai, and Shijun Liao.

This work is supported in part by the Natural Science Foundation of China No.~NFSC61073087, the National 973 Projects No.~2010CB529200, and the Natural Science Foundation of China No.~NFSC91029738.


\begin{thebibliography}{42}%
\makeatletter
\providecommand \@ifxundefined [1]{%
 \@ifx{#1\undefined}
}%
\providecommand \@ifnum [1]{%
 \ifnum #1\expandafter \@firstoftwo
 \else \expandafter \@secondoftwo
 \fi
}%
\providecommand \@ifx [1]{%
 \ifx #1\expandafter \@firstoftwo
 \else \expandafter \@secondoftwo
 \fi
}%
\providecommand \natexlab [1]{#1}%
\providecommand \enquote  [1]{``#1''}%
\providecommand \bibnamefont  [1]{#1}%
\providecommand \bibfnamefont [1]{#1}%
\providecommand \citenamefont [1]{#1}%
\providecommand \href@noop [0]{\@secondoftwo}%
\providecommand \href [0]{\begingroup \@sanitize@url \@href}%
\providecommand \@href[1]{\@@startlink{#1}\@@href}%
\providecommand \@@href[1]{\endgroup#1\@@endlink}%
\providecommand \@sanitize@url [0]{\catcode `\\12\catcode `\$12\catcode
  `\&12\catcode `\#12\catcode `\^12\catcode `\_12\catcode `\%12\relax}%
\providecommand \@@startlink[1]{}%
\providecommand \@@endlink[0]{}%
\providecommand \url  [0]{\begingroup\@sanitize@url \@url }%
\providecommand \@url [1]{\endgroup\@href {#1}{\urlprefix }}%
\providecommand \urlprefix  [0]{URL }%
\providecommand \Eprint [0]{\href }%
\providecommand \doibase [0]{http://dx.doi.org/}%
\providecommand \selectlanguage [0]{\@gobble}%
\providecommand \bibinfo  [0]{\@secondoftwo}%
\providecommand \bibfield  [0]{\@secondoftwo}%
\providecommand \translation [1]{[#1]}%
\providecommand \BibitemOpen [0]{}%
\providecommand \bibitemStop [0]{}%
\providecommand \bibitemNoStop [0]{.\EOS\space}%
\providecommand \EOS [0]{\spacefactor3000\relax}%
\providecommand \BibitemShut  [1]{\csname bibitem#1\endcsname}%
\let\auto@bib@innerbib\@empty
\bibitem [{\citenamefont {Strogatz}(2000)}]{strogatz2000nonlinear}%
  \BibitemOpen
  \bibfield  {author} {\bibinfo {author} {\bibfnamefont {S.~H.}\ \bibnamefont
  {Strogatz}},\ }\href@noop {} {\emph {\bibinfo {title} {{Nonlinear Dynamics
  and Chaos: With Applications to Physics, Biology, Chemistry, and
  Engineering}}}}\ (\bibinfo  {publisher} {Perseus Books},\ \bibinfo {address}
  {Reading},\ \bibinfo {year} {2000})\ p.\ \bibinfo {pages} {201}\BibitemShut
  {NoStop}%
\bibitem [{\citenamefont {Rice}(2004)}]{rice}%
  \BibitemOpen
  \bibfield  {author} {\bibinfo {author} {\bibfnamefont {S.~H.}\ \bibnamefont
  {Rice}},\ }\href@noop {} {\emph {\bibinfo {title} {{Evolutionary Theory:
  Mathematical and Conceptual Foundations}}}}\ (\bibinfo  {publisher} {Sinauer
  Associates},\ \bibinfo {address} {Sunderland},\ \bibinfo {year}
  {2004})\BibitemShut {NoStop}%
\bibitem [{\citenamefont {L\"{u}}\ \emph {et~al.}(2002)\citenamefont {L\"{u}},
  \citenamefont {Zhou}, \citenamefont {Chen},\ and\ \citenamefont
  {Yang}}]{Chen02generateAttractor}%
  \BibitemOpen
  \bibfield  {author} {\bibinfo {author} {\bibfnamefont {J.}~\bibnamefont
  {L\"{u}}}, \bibinfo {author} {\bibfnamefont {T.}~\bibnamefont {Zhou}},
  \bibinfo {author} {\bibfnamefont {G.}~\bibnamefont {Chen}}, \ and\ \bibinfo
  {author} {\bibfnamefont {X.}~\bibnamefont {Yang}},\ }\href@noop {} {\bibfield
   {journal} {\bibinfo  {journal} {Chaos}\ }\textbf {\bibinfo {volume} {12}},\
  \bibinfo {pages} {344} (\bibinfo {year} {2002})}\BibitemShut {NoStop}%
\bibitem [{\citenamefont {Poincar\'e}(1914)}]{Poincare}%
  \BibitemOpen
  \bibfield  {author} {\bibinfo {author} {\bibfnamefont {H.~J.}\ \bibnamefont
  {Poincar\'e}},\ }\href@noop {} {\emph {\bibinfo {title} {{Science and Method,
  Translated by F. Maitland}}}}\ (\bibinfo  {publisher} {T. Nelson and Sons},\
  \bibinfo {address} {New York},\ \bibinfo {year} {1914})\BibitemShut {NoStop}%
\bibitem [{\citenamefont {Strogatz}(2001)}]{Strogatz}%
  \BibitemOpen
  \bibfield  {author} {\bibinfo {author} {\bibfnamefont {S.~H.}\ \bibnamefont
  {Strogatz}},\ }\href@noop {} {\bibfield  {journal} {\bibinfo  {journal}
  {Nature}\ }\textbf {\bibinfo {volume} {410}},\ \bibinfo {pages} {268}
  (\bibinfo {year} {2001})}\BibitemShut {NoStop}%
\bibitem [{\citenamefont {Zhang}\ \emph {et~al.}(2012)\citenamefont {Zhang},
  \citenamefont {Qian},\ and\ \citenamefont {Qian}}]{Qian12Stochastic}%
  \BibitemOpen
  \bibfield  {author} {\bibinfo {author} {\bibfnamefont {X.-J.}\ \bibnamefont
  {Zhang}}, \bibinfo {author} {\bibfnamefont {H.}~\bibnamefont {Qian}}, \ and\
  \bibinfo {author} {\bibfnamefont {M.}~\bibnamefont {Qian}},\ }\href@noop {}
  {\bibfield  {journal} {\bibinfo  {journal} {Physics Reports}\ }\textbf
  {\bibinfo {volume} {510}},\ \bibinfo {pages} {1 } (\bibinfo {year}
  {2012})}\BibitemShut {NoStop}%
\bibitem [{\citenamefont {Thom}(1975)}]{Thom}%
  \BibitemOpen
  \bibfield  {author} {\bibinfo {author} {\bibfnamefont {R.}~\bibnamefont
  {Thom}},\ }\href@noop {} {\emph {\bibinfo {title} {{Structural Stability and
  Morphogenesis: An Outline of a General Theory of Models, Translated by D.H.
  Fowler}}}}\ (\bibinfo  {publisher} {Benjamin},\ \bibinfo {address}
  {Reading},\ \bibinfo {year} {1975})\BibitemShut {NoStop}%
\bibitem [{\citenamefont {Zeeman}(1976)}]{Zeeman}%
  \BibitemOpen
  \bibfield  {author} {\bibinfo {author} {\bibfnamefont {E.~C.}\ \bibnamefont
  {Zeeman}},\ }\href@noop {} {\bibfield  {journal} {\bibinfo  {journal} {Sci.
  Am.}\ }\textbf {\bibinfo {volume} {234}},\ \bibinfo {pages} {65} (\bibinfo
  {year} {1976})}\BibitemShut {NoStop}%
\bibitem [{\citenamefont {Sol\'{e}}\ \emph {et~al.}(1992)\citenamefont
  {Sol\'{e}}, \citenamefont {Bascompte},\ and\ \citenamefont
  {Valls}}]{Sole92Stability}%
  \BibitemOpen
  \bibfield  {author} {\bibinfo {author} {\bibfnamefont {R.~V.}\ \bibnamefont
  {Sol\'{e}}}, \bibinfo {author} {\bibfnamefont {J.}~\bibnamefont {Bascompte}},
  \ and\ \bibinfo {author} {\bibfnamefont {J.}~\bibnamefont {Valls}},\
  }\href@noop {} {\bibfield  {journal} {\bibinfo  {journal} {Chaos}\ }\textbf
  {\bibinfo {volume} {2}},\ \bibinfo {pages} {387} (\bibinfo {year}
  {1992})}\BibitemShut {NoStop}%
\bibitem [{\citenamefont {Breymann}\ \emph {et~al.}(1998)\citenamefont
  {Breymann}, \citenamefont {T\'{e}l},\ and\ \citenamefont
  {Vollmer}}]{Breymann98Reversibility}%
  \BibitemOpen
  \bibfield  {author} {\bibinfo {author} {\bibfnamefont {W.}~\bibnamefont
  {Breymann}}, \bibinfo {author} {\bibfnamefont {T.}~\bibnamefont {T\'{e}l}}, \
  and\ \bibinfo {author} {\bibfnamefont {J.}~\bibnamefont {Vollmer}},\
  }\href@noop {} {\bibfield  {journal} {\bibinfo  {journal} {Chaos}\ }\textbf
  {\bibinfo {volume} {8}},\ \bibinfo {pages} {396} (\bibinfo {year}
  {1998})}\BibitemShut {NoStop}%
\bibitem [{\citenamefont {Smale}(1978)}]{Smale}%
  \BibitemOpen
  \bibfield  {author} {\bibinfo {author} {\bibfnamefont {S.}~\bibnamefont
  {Smale}},\ }\href@noop {} {\bibfield  {journal} {\bibinfo  {journal} {Bull.
  Amer. Math. Soc.}\ }\textbf {\bibinfo {volume} {84}},\ \bibinfo {pages}
  {1360} (\bibinfo {year} {1978})}\BibitemShut {NoStop}%
\bibitem [{\citenamefont {Zhu}\ \emph {et~al.}(2006)\citenamefont {Zhu},
  \citenamefont {Yin},\ and\ \citenamefont {Ao}}]{Ao06limitcycle}%
  \BibitemOpen
  \bibfield  {author} {\bibinfo {author} {\bibfnamefont {X.-M.}\ \bibnamefont
  {Zhu}}, \bibinfo {author} {\bibfnamefont {L.}~\bibnamefont {Yin}}, \ and\
  \bibinfo {author} {\bibfnamefont {P.}~\bibnamefont {Ao}},\ }\href@noop {}
  {\bibfield  {journal} {\bibinfo  {journal} {Int. J. Mod. Phys. B}\ }\textbf
  {\bibinfo {volume} {20}},\ \bibinfo {pages} {817} (\bibinfo {year}
  {2006})}\BibitemShut {NoStop}%
\bibitem [{\citenamefont {Ma}\ \emph {et~al.}(2012)\citenamefont {Ma},
  \citenamefont {Yuan}, \citenamefont {Li}, \citenamefont {Ao},\ and\
  \citenamefont {Yuan}}]{Yian}%
  \BibitemOpen
  \bibfield  {author} {\bibinfo {author} {\bibfnamefont {Y.}~\bibnamefont
  {Ma}}, \bibinfo {author} {\bibfnamefont {R.}~\bibnamefont {Yuan}}, \bibinfo
  {author} {\bibfnamefont {Y.}~\bibnamefont {Li}}, \bibinfo {author}
  {\bibfnamefont {P.}~\bibnamefont {Ao}}, \ and\ \bibinfo {author}
  {\bibfnamefont {B.}~\bibnamefont {Yuan}},\ }\href@noop {} {\enquote {\bibinfo
  {title} {{{Lyapunov functions in piecewise linear systems: From fixed point
  to limit cycle}}},}\ } (\bibinfo {year} {2012}),\ \bibinfo {note}
  {submitted}\BibitemShut {NoStop}%
\bibitem [{\citenamefont {Sira-Ramirez}\ and\ \citenamefont
  {Cruz-Hernandez}(2001)}]{Sira}%
  \BibitemOpen
  \bibfield  {author} {\bibinfo {author} {\bibfnamefont {H.}~\bibnamefont
  {Sira-Ramirez}}\ and\ \bibinfo {author} {\bibfnamefont {C.}~\bibnamefont
  {Cruz-Hernandez}},\ }\href@noop {} {\bibfield  {journal} {\bibinfo  {journal}
  {Int. J. Bifurcation and Chaos}\ }\textbf {\bibinfo {volume} {11}},\ \bibinfo
  {pages} {1381} (\bibinfo {year} {2001})}\BibitemShut {NoStop}%
\bibitem [{\citenamefont {Sarasola}\ \emph {et~al.}(2005)\citenamefont
  {Sarasola}, \citenamefont {d'Anjou}, \citenamefont {Torrealdea},\ and\
  \citenamefont {Moujahid}}]{Sarasola}%
  \BibitemOpen
  \bibfield  {author} {\bibinfo {author} {\bibfnamefont {C.}~\bibnamefont
  {Sarasola}}, \bibinfo {author} {\bibfnamefont {A.}~\bibnamefont {d'Anjou}},
  \bibinfo {author} {\bibfnamefont {F.~J.}\ \bibnamefont {Torrealdea}}, \ and\
  \bibinfo {author} {\bibfnamefont {A.}~\bibnamefont {Moujahid}},\ }\href@noop
  {} {\bibfield  {journal} {\bibinfo  {journal} {Int. J. Bifurcation and
  Chaos}\ }\textbf {\bibinfo {volume} {15}},\ \bibinfo {pages} {2507} (\bibinfo
  {year} {2005})}\BibitemShut {NoStop}%
\bibitem [{\citenamefont {Zhou}\ and\ \citenamefont {E}(2010)}]{Weinan}%
  \BibitemOpen
  \bibfield  {author} {\bibinfo {author} {\bibfnamefont {X.}~\bibnamefont
  {Zhou}}\ and\ \bibinfo {author} {\bibfnamefont {W.}~\bibnamefont {E}},\
  }\href@noop {} {\bibfield  {journal} {\bibinfo  {journal} {Comm. Math. Sci.}\
  }\textbf {\bibinfo {volume} {8}},\ \bibinfo {pages} {341} (\bibinfo {year}
  {2010})}\BibitemShut {NoStop}%
\bibitem [{\citenamefont {Lorenz}(1963)}]{Lorenz}%
  \BibitemOpen
  \bibfield  {author} {\bibinfo {author} {\bibfnamefont {E.~N.}\ \bibnamefont
  {Lorenz}},\ }\href@noop {} {\bibfield  {journal} {\bibinfo  {journal} {J.
  Atmosph. Sci.}\ }\textbf {\bibinfo {volume} {20}},\ \bibinfo {pages} {130}
  (\bibinfo {year} {1963})}\BibitemShut {NoStop}%
\bibitem [{\citenamefont {Ao}(2004)}]{ao04potential}%
  \BibitemOpen
  \bibfield  {author} {\bibinfo {author} {\bibfnamefont {P.}~\bibnamefont
  {Ao}},\ }\href@noop {} {\bibfield  {journal} {\bibinfo  {journal} {J. Phys.
  A: Math. Gen.}\ }\textbf {\bibinfo {volume} {37}},\ \bibinfo {pages} {25}
  (\bibinfo {year} {2004})}\BibitemShut {NoStop}%
\bibitem [{\citenamefont {Yuan}\ \emph {et~al.}(2011)\citenamefont {Yuan},
  \citenamefont {Ma}, \citenamefont {Yuan},\ and\ \citenamefont {Ao}}]{Ruoshi}%
  \BibitemOpen
  \bibfield  {author} {\bibinfo {author} {\bibfnamefont {R.}~\bibnamefont
  {Yuan}}, \bibinfo {author} {\bibfnamefont {Y.}~\bibnamefont {Ma}}, \bibinfo
  {author} {\bibfnamefont {B.}~\bibnamefont {Yuan}}, \ and\ \bibinfo {author}
  {\bibfnamefont {P.}~\bibnamefont {Ao}},\ }in\ \href@noop {} {\emph {\bibinfo
  {booktitle} {Proceedings of 30th Chinese Control Conference (CCC), 2011}}}\
  (\bibinfo {organization} {IEEE},\ \bibinfo {year} {2011})\ pp.\ \bibinfo
  {pages} {6573--6580}\BibitemShut {NoStop}%
\bibitem [{\citenamefont {Hirsch}\ \emph {et~al.}(2004)\citenamefont {Hirsch},
  \citenamefont {Smale},\ and\ \citenamefont {Devaney}}]{hirsch}%
  \BibitemOpen
  \bibfield  {author} {\bibinfo {author} {\bibfnamefont {M.~W.}\ \bibnamefont
  {Hirsch}}, \bibinfo {author} {\bibfnamefont {S.}~\bibnamefont {Smale}}, \
  and\ \bibinfo {author} {\bibfnamefont {R.~L.}\ \bibnamefont {Devaney}},\
  }\href@noop {} {\emph {\bibinfo {title} {{Differential Equations, Dynamical
  Systems, and an Introduction to Chaos}}}}\ (\bibinfo  {publisher}
  {Elsevier/Academic Press},\ \bibinfo {address} {Amsterdam},\ \bibinfo {year}
  {2004})\BibitemShut {NoStop}%
\bibitem [{\citenamefont {Shilnikov}\ \emph {et~al.}(1998)\citenamefont
  {Shilnikov}, \citenamefont {Shilnikov}, \citenamefont {Turaev},\ and\
  \citenamefont {Chua}}]{Shilnikov}%
  \BibitemOpen
  \bibfield  {author} {\bibinfo {author} {\bibfnamefont {L.~P.}\ \bibnamefont
  {Shilnikov}}, \bibinfo {author} {\bibfnamefont {A.~L.}\ \bibnamefont
  {Shilnikov}}, \bibinfo {author} {\bibfnamefont {D.~V.}\ \bibnamefont
  {Turaev}}, \ and\ \bibinfo {author} {\bibfnamefont {L.~O.}\ \bibnamefont
  {Chua}},\ }\href@noop {} {\emph {\bibinfo {title} {{Methods of Qualitative
  Theory in Nonlinear Dynamics, Part I}}}}\ (\bibinfo  {publisher} {World
  Scientific},\ \bibinfo {address} {Singapore},\ \bibinfo {year}
  {1998})\BibitemShut {NoStop}%
\bibitem [{\citenamefont {Robinson}(2004)}]{robinson}%
  \BibitemOpen
  \bibfield  {author} {\bibinfo {author} {\bibfnamefont {R.}~\bibnamefont
  {Robinson}},\ }\href@noop {} {\emph {\bibinfo {title} {An Introduction to
  Dynamical Systems: Continuous and Discrete}}}\ (\bibinfo  {publisher}
  {Pearson Prentice Hall},\ \bibinfo {address} {New Jersey},\ \bibinfo {year}
  {2004})\BibitemShut {NoStop}%
\bibitem [{\citenamefont {Cheng}\ \emph {et~al.}(2000)\citenamefont {Cheng},
  \citenamefont {Spurgeon},\ and\ \citenamefont {Xiang}}]{Daizhan}%
  \BibitemOpen
  \bibfield  {author} {\bibinfo {author} {\bibfnamefont {D.}~\bibnamefont
  {Cheng}}, \bibinfo {author} {\bibfnamefont {S.}~\bibnamefont {Spurgeon}}, \
  and\ \bibinfo {author} {\bibfnamefont {J.}~\bibnamefont {Xiang}},\ }in\
  \href@noop {} {\emph {\bibinfo {booktitle} {Proceedings of the 39th IEEE
  Conference on Decision and Control, 2000}}},\ Vol.~\bibinfo {volume} {5}\
  (\bibinfo {organization} {IEEE},\ \bibinfo {year} {2000})\ pp.\ \bibinfo
  {pages} {5125--5130}\BibitemShut {NoStop}%
\bibitem [{\citenamefont {Li}\ \emph {et~al.}(2012)\citenamefont {Li},
  \citenamefont {Wang},\ and\ \citenamefont {Wang}}]{Wangjin}%
  \BibitemOpen
  \bibfield  {author} {\bibinfo {author} {\bibfnamefont {C.}~\bibnamefont
  {Li}}, \bibinfo {author} {\bibfnamefont {E.}~\bibnamefont {Wang}}, \ and\
  \bibinfo {author} {\bibfnamefont {J.}~\bibnamefont {Wang}},\ }\href@noop {}
  {\bibfield  {journal} {\bibinfo  {journal} {The Journal of Chemical Physics}\
  }\textbf {\bibinfo {volume} {136}},\ \bibinfo {eid} {194108} (\bibinfo {year}
  {2012})}\BibitemShut {NoStop}%
\bibitem [{\citenamefont {Kobe}(1986)}]{Helmholtz}%
  \BibitemOpen
  \bibfield  {author} {\bibinfo {author} {\bibfnamefont {D.~H.}\ \bibnamefont
  {Kobe}},\ }\href@noop {} {\bibfield  {journal} {\bibinfo  {journal} {Amer. J.
  Phys.}\ }\textbf {\bibinfo {volume} {54}},\ \bibinfo {pages} {552} (\bibinfo
  {year} {1986})}\BibitemShut {NoStop}%
\bibitem [{\citenamefont {Olson}\ and\ \citenamefont {Ao}(2007)}]{Ao07Decomp}%
  \BibitemOpen
  \bibfield  {author} {\bibinfo {author} {\bibfnamefont {J.~C.}\ \bibnamefont
  {Olson}}\ and\ \bibinfo {author} {\bibfnamefont {P.}~\bibnamefont {Ao}},\
  }\href@noop {} {\bibfield  {journal} {\bibinfo  {journal} {Phys. Rev. B}\
  }\textbf {\bibinfo {volume} {75}},\ \bibinfo {pages} {035114} (\bibinfo
  {year} {2007})}\BibitemShut {NoStop}%
\bibitem [{\citenamefont {Arnold}\ \emph {et~al.}(1989)\citenamefont {Arnold},
  \citenamefont {Weinstein},\ and\ \citenamefont {Vogtmann}}]{Arnold89Math}%
  \BibitemOpen
  \bibfield  {author} {\bibinfo {author} {\bibfnamefont {V.}~\bibnamefont
  {Arnold}}, \bibinfo {author} {\bibfnamefont {A.}~\bibnamefont {Weinstein}}, \
  and\ \bibinfo {author} {\bibfnamefont {K.}~\bibnamefont {Vogtmann}},\
  }\href@noop {} {\emph {\bibinfo {title} {Mathematical Methods of Classical
  Mechanics}}},\ \bibinfo {edition} {2nd}\ ed.\ (\bibinfo  {publisher}
  {Springer Verlag},\ \bibinfo {address} {Berlin},\ \bibinfo {year}
  {1989})\BibitemShut {NoStop}%
\bibitem [{Note1()}]{Note1}%
  \BibitemOpen
  \bibinfo {note} {To avoid confusion, we restrict the use of generalized
  Poisson brackets in this section (section $2$)}\BibitemShut {NoStop}%
\bibitem [{\citenamefont {Binder}\ and\ \citenamefont
  {Laverde}(1999)}]{Binder99Numeric}%
  \BibitemOpen
  \bibfield  {author} {\bibinfo {author} {\bibfnamefont {P.-M.}\ \bibnamefont
  {Binder}}\ and\ \bibinfo {author} {\bibfnamefont {D.}~\bibnamefont
  {Laverde}},\ }\href@noop {} {\bibfield  {journal} {\bibinfo  {journal}
  {Chaos}\ }\textbf {\bibinfo {volume} {9}},\ \bibinfo {pages} {206} (\bibinfo
  {year} {1999})}\BibitemShut {NoStop}%
\bibitem [{\citenamefont {Tucker}(1999)}]{Tucker}%
  \BibitemOpen
  \bibfield  {author} {\bibinfo {author} {\bibfnamefont {W.}~\bibnamefont
  {Tucker}},\ }\href@noop {} {\bibfield  {journal} {\bibinfo  {journal} {C. R.
  Acad. Sci. Paris}\ }\textbf {\bibinfo {volume} {328}},\ \bibinfo {pages}
  {1197} (\bibinfo {year} {1999})}\BibitemShut {NoStop}%
\bibitem [{\citenamefont {Smale}(1998)}]{SmaleCentury}%
  \BibitemOpen
  \bibfield  {author} {\bibinfo {author} {\bibfnamefont {S.}~\bibnamefont
  {Smale}},\ }\href@noop {} {\bibfield  {journal} {\bibinfo  {journal} {Math.
  Intelligencer}\ }\textbf {\bibinfo {volume} {20}},\ \bibinfo {pages} {7}
  (\bibinfo {year} {1998})}\BibitemShut {NoStop}%
\bibitem [{\citenamefont {Guckenheimer}\ and\ \citenamefont
  {Williams}(1979)}]{John}%
  \BibitemOpen
  \bibfield  {author} {\bibinfo {author} {\bibfnamefont {J.}~\bibnamefont
  {Guckenheimer}}\ and\ \bibinfo {author} {\bibfnamefont {R.~F.}\ \bibnamefont
  {Williams}},\ }\href@noop {} {\bibfield  {journal} {\bibinfo  {journal}
  {Publ. Math. IHES}\ }\textbf {\bibinfo {volume} {50}},\ \bibinfo {pages} {59}
  (\bibinfo {year} {1979})}\BibitemShut {NoStop}%
\bibitem [{Note2()}]{Note2}%
  \BibitemOpen
  \bibinfo {note} {The square brackets in this (section $3$) and the following
  sections mean closed intervals, not the generalized Poisson
  brackets}\BibitemShut {NoStop}%
\bibitem [{\citenamefont {Kuznetsov}(2012)}]{Kuznetsov}%
  \BibitemOpen
  \bibfield  {author} {\bibinfo {author} {\bibfnamefont {S.~P.}\ \bibnamefont
  {Kuznetsov}},\ }\href@noop {} {\emph {\bibinfo {title} {{Hyperbolic Chaos: A
  Physicist's View}}}}\ (\bibinfo  {publisher} {Springer Verlag},\ \bibinfo
  {address} {Berlin, Heidelberg},\ \bibinfo {year} {2012})\BibitemShut
  {NoStop}%
\bibitem [{\citenamefont {Peitgen}\ \emph {et~al.}(2004)\citenamefont
  {Peitgen}, \citenamefont {J\"urgens},\ and\ \citenamefont
  {Saupe}}]{frontier}%
  \BibitemOpen
  \bibfield  {author} {\bibinfo {author} {\bibfnamefont {H.-O.}\ \bibnamefont
  {Peitgen}}, \bibinfo {author} {\bibfnamefont {H.}~\bibnamefont {J\"urgens}},
  \ and\ \bibinfo {author} {\bibfnamefont {D.}~\bibnamefont {Saupe}},\
  }\href@noop {} {\emph {\bibinfo {title} {Chaos and Fractals: New Frontiers of
  Science}}}\ (\bibinfo  {publisher} {Springer Verlag},\ \bibinfo {address}
  {New York},\ \bibinfo {year} {2004})\BibitemShut {NoStop}%
\bibitem [{\citenamefont {Anishchenko}\ and\ \citenamefont
  {Strelkova}(1998)}]{NonStrangeChaotic}%
  \BibitemOpen
  \bibfield  {author} {\bibinfo {author} {\bibfnamefont {V.~S.}\ \bibnamefont
  {Anishchenko}}\ and\ \bibinfo {author} {\bibfnamefont {G.~I.}\ \bibnamefont
  {Strelkova}},\ }\href@noop {} {\bibfield  {journal} {\bibinfo  {journal}
  {Discrete Dyn. Nat. Soc.}\ }\textbf {\bibinfo {volume} {2}},\ \bibinfo
  {pages} {53} (\bibinfo {year} {1998})}\BibitemShut {NoStop}%
\bibitem [{\citenamefont {Arnold}(1983)}]{Arnold83Geo}%
  \BibitemOpen
  \bibfield  {author} {\bibinfo {author} {\bibfnamefont {V.}~\bibnamefont
  {Arnold}},\ }\href@noop {} {\emph {\bibinfo {title} {Geometrical Methods in
  the Theory of Ordinary Differential Equations}}}\ (\bibinfo  {publisher}
  {Springer Verlag},\ \bibinfo {address} {New York},\ \bibinfo {year}
  {1983})\BibitemShut {NoStop}%
\bibitem [{\citenamefont {Freidlin}\ and\ \citenamefont
  {Wentzell}(2008)}]{WentzellRecent}%
  \BibitemOpen
  \bibfield  {author} {\bibinfo {author} {\bibfnamefont {M.}~\bibnamefont
  {Freidlin}}\ and\ \bibinfo {author} {\bibfnamefont {A.}~\bibnamefont
  {Wentzell}},\ }in\ \href@noop {} {\emph {\bibinfo {booktitle} {Topics in
  Stochastic Analysis and Nonparametric Estimation}}},\ \bibinfo {series} {The
  IMA Volumes in Mathematics and its Applications}, Vol.\ \bibinfo {volume}
  {145},\ \bibinfo {editor} {edited by\ \bibinfo {editor} {\bibfnamefont
  {P.-L.}\ \bibnamefont {Chow}}, \bibinfo {editor} {\bibfnamefont
  {G.}~\bibnamefont {Yin}}, \ and\ \bibinfo {editor} {\bibfnamefont
  {B.}~\bibnamefont {Mordukhovich}}}\ (\bibinfo  {publisher} {Springer
  Verlag},\ \bibinfo {address} {New York},\ \bibinfo {year} {2008})\ pp.\
  \bibinfo {pages} {1--19}\BibitemShut {NoStop}%
\bibitem [{\citenamefont {Freidlin}\ and\ \citenamefont
  {Wentzell}(1998)}]{Wentzell}%
  \BibitemOpen
  \bibfield  {author} {\bibinfo {author} {\bibfnamefont {M.~I.}\ \bibnamefont
  {Freidlin}}\ and\ \bibinfo {author} {\bibfnamefont {A.~D.}\ \bibnamefont
  {Wentzell}},\ }\href@noop {} {\emph {\bibinfo {title} {{Random Perturbations
  of Dynamical Systems}}}},\ \bibinfo {edition} {2nd}\ ed.,\ Grundlehren der
  mathematischen Wissenschaften\ (\bibinfo  {publisher} {Springer Verlag},\
  \bibinfo {address} {New York},\ \bibinfo {year} {1998})\BibitemShut {NoStop}%
\bibitem [{\citenamefont {Gilmore}(1998)}]{Gilmore}%
  \BibitemOpen
  \bibfield  {author} {\bibinfo {author} {\bibfnamefont {R.}~\bibnamefont
  {Gilmore}},\ }\href@noop {} {\bibfield  {journal} {\bibinfo  {journal} {Rev.
  Mod. Phys.}\ }\textbf {\bibinfo {volume} {70}},\ \bibinfo {pages} {1455}
  (\bibinfo {year} {1998})}\BibitemShut {NoStop}%
\bibitem [{\citenamefont {Grebogi}\ \emph {et~al.}(1984)\citenamefont
  {Grebogi}, \citenamefont {Ott}, \citenamefont {Pelican},\ and\ \citenamefont
  {Yorke}}]{StrangeAndChaotic}%
  \BibitemOpen
  \bibfield  {author} {\bibinfo {author} {\bibfnamefont {C.}~\bibnamefont
  {Grebogi}}, \bibinfo {author} {\bibfnamefont {E.}~\bibnamefont {Ott}},
  \bibinfo {author} {\bibfnamefont {S.}~\bibnamefont {Pelican}}, \ and\
  \bibinfo {author} {\bibfnamefont {J.}~\bibnamefont {Yorke}},\ }\href@noop {}
  {\bibfield  {journal} {\bibinfo  {journal} {Physica D}\ }\textbf {\bibinfo
  {volume} {13}},\ \bibinfo {pages} {261} (\bibinfo {year} {1984})}\BibitemShut
  {NoStop}%
\bibitem [{\citenamefont {Prasad}\ \emph {et~al.}(2007)\citenamefont {Prasad},
  \citenamefont {Nandi},\ and\ \citenamefont {Ramaswamy}}]{StrangeNonChaotic}%
  \BibitemOpen
  \bibfield  {author} {\bibinfo {author} {\bibfnamefont {A.}~\bibnamefont
  {Prasad}}, \bibinfo {author} {\bibfnamefont {A.}~\bibnamefont {Nandi}}, \
  and\ \bibinfo {author} {\bibfnamefont {R.}~\bibnamefont {Ramaswamy}},\
  }\href@noop {} {\bibfield  {journal} {\bibinfo  {journal} {Int. J.
  Bifurcation and Chaos}\ }\textbf {\bibinfo {volume} {17}},\ \bibinfo {pages}
  {3397} (\bibinfo {year} {2007})}\BibitemShut {NoStop}%
\end{thebibliography}

%

\end{document}